\documentclass{amsproc}
\usepackage[margin=1in]{geometry} 
\usepackage{amsmath,amsthm,amssymb}
\usepackage{pifont}
\usepackage{mathrsfs}
\usepackage[all,cmtip]{xy}
\usepackage{graphics,graphicx}

\usepackage{tcolorbox}

\theoremstyle{plain}
 \newtheorem{thm}{Theorem}[section]
 \newtheorem{prop}[thm]{Proposition}
 \newtheorem{lem}[thm]{Lemma}
 \newtheorem{cor}[thm]{Corollary}
\theoremstyle{definition}
 \newtheorem{exm}{Example}[section]
 \newtheorem{dfn}{Definition}[section]
 \newtheorem{nota}{Notation}[section]
\theoremstyle{remark}
 \newtheorem{rem}{Remark}[section]
 \numberwithin{equation}{section}

\renewcommand{\leq}{\leqslant}
\renewcommand{\geq}{\geqslant}

\usepackage{hyperref} 
\hypersetup{colorlinks=true,
    linkcolor=blue,         
    citecolor=green,        
    filecolor=magenta,      
    urlcolor=cyan           
}

\title{Enumerative Methods in Quantum Electrodynamics}

\begin{document}

\maketitle{}
\vskip 0.4cm
\begin{center}
    Ali Assem Mahmoud\\
National Research Council of Canada, Digital Technologies\\Waterloo, ON, Canada.\\
\href{mailto:ali.mahmoud@nrc-cnrc.gc.ca}{\tt ali.mahmoud@uwaterloo.ca}
\end{center}

\vspace{0.5cm}
\begin{abstract}
We show that observables in QED-type theories can be realized in terms of a combinatorial structure called chord diagrams. One advantage of this combinatorial representation is that it simplifies the study of the asymptotic behavior of corresponding Green functions. Particularly, using the new representation, there is no need to use the standard approach of singularity analysis. This relation also reveals the unexplained correlation between the number of Feynman diagrams in Yukawa theory and the diagrams in quenched QED. 
\end{abstract}

\maketitle



\section{Introduction and Setup }\label{quenchedsec}

In perturbation theory, it is essential to study the behaviour of Green functions at large orders. This however usually depends on the problem of approximating the number of diagrams contributing to the perturbative terms, and, in that regards, it is shown in \cite{numandweights} that it  suffices to consider  the problem as a zero-dimensional field theory. 
In this paper we will establish a strong relation between 1PI Feynman diagrams in certain QED-type theories and a combinatorial structure called \textit{chord diagrams}. Chord diagrams offer more simplicity when studying the respective generating functions at large orders. Moreover, the way these Feynman diagrams are mapped into chord diagrams is cryptic, which should probably mean that there is still more to highlight in this direction. This also gives rise to the question of why these graph-theoretic matchings are sufficient to represent the theory.  The method used here is generic and can be applied to other QFT theories with a cubic-type vertex. We will regularly appeal to theorems about factorially divergent power series (see \cite{michi1} or Appendix A in \cite{alipaper2con}). First we briefly set-up the context in perturbation theory. In most parts we follow the notation in \cite{michiq} as we will need to compare some of the results eventually.

\subsection{Zero-Dimensional Scalar Theories with Interaction}

We are particularly considering quenched QED and Yukawa theory. These are examples of theories with interaction, where the partition function takes the form 

\[Z(\hbar,j):=\int_\mathbb{R} \displaystyle\frac{1}{\sqrt{2\pi\hbar}} e^{\frac{1}{\hbar}\left(-\frac{x^2}{2}+V(x)+xj\right)}dx,\] 
where an additional term is added to the potential, namely $xj$, with $j$ called the \textit{source}.

 The generated Feynman diagrams are labeled, and hence in order to restrict ourselves to connected diagrams we have to take the logarithm of the partition function, commonly known as the \textit{free energy}:
\begin{align*}
    W(\hbar,j):&=\hbar \log 
    Z(\hbar,j),
\end{align*}
which generates all connected diagrams.

As customary in QFT, to move to the \textit{quantum effective action} $G$, which generates \textit{1PI} diagrams, one takes the Legendre transform of $W$:

\begin{align} \label{propergreenfnmichi}
    G(\hbar,\varphi_c):=W-j \varphi_c, 
\end{align}
where $\varphi_c:=\partial_jW$. The coefficients $[\varphi_c^n]G$ are called the \textit{(proper) Green functions} of the theory. 
Recall that from a graph-theoretic point of view, being \textit{1PI} (1-particle irreducible)  is another way of saying $2$-connected. Thus, combinatorially, the Legendre transform, as in \cite{kjm2}, is seen to be the transportation from connected diagrams to $2$-connected or \textit{1PI} diagrams. In that sense, the order of the derivative  $\partial^n_{\varphi_c}G|_{\varphi_c=0}=[\varphi_c^n]G$ determines the number of external legs.

In the next part of the discussion, we shall need the following physical jargon and terminology:

\begin{enumerate}
    \item The Green function $[\varphi_c^1]G=\partial_{\varphi_c}G|_{\varphi_c=0}$ generates all \textit{1PI} diagrams with exactly one external leg, which are called the \textit{tadpoles} of the theory (Figure \ref{tp}).
    
    \begin{figure}[h]
    \centering
    \includegraphics[scale=0.34]{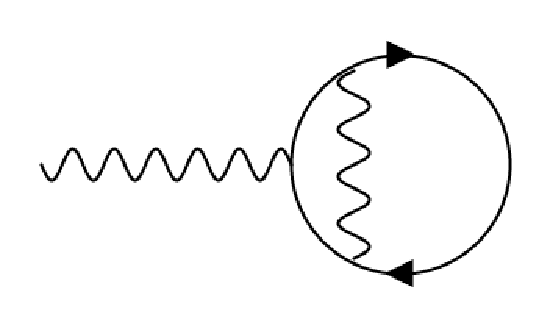}
    \caption{A tadpole diagram in QED}
    \label{tp}\end{figure}
    
    \item The Green function $[\varphi^2_c]G=\partial^2_{\varphi_c}G|_{\varphi_c=0}$ generates all \textit{1PI} diagrams with two external legs. Such a diagram is called a \textit{1PI propagator} (can replace an edge in the theory; Figure \ref{propagator}). 

    \begin{figure}[h]
    \centering
    \includegraphics[scale=0.34]{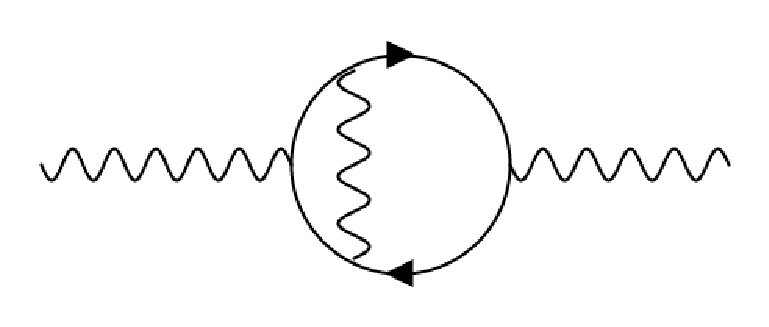}
    \caption{A propagator diagram}
    \label{propagator}
    \end{figure}
 
    \item For $n>2$, $\partial^n_{\varphi_c}G|_{\varphi_c=0}=[\varphi_c^n]G$ is called the \textit{$n$-point function}.
\end{enumerate}

In quenched QED, some of the quantities that we are going to compare their expansions with the generating series of $2$-connected chord diagrams are the \textit{renormalized} Green functions with respect to a chosen \textit{residue}. We will not however digress into the  Hopf-algebraic treatment of renormalization here. For more about this topic the reader can consult \cite{manchonhopf,alipaper2con}, or the original paper by D. Kreimer and A. Connes \cite{conneskreimer}.



\section{$k$-Connected Chord Diagrams}

\begin{dfn}[Chord diagrams]
A \textit{chord diagram} on $n$ chords (i.e. of size $n$) is geometrically perceived as a circle with $2n$ nodes that are matched into disjoint pairs, with each pair corresponding to a \textit{chord}. 
\end{dfn}

\begin{dfn}[Rooted chord diagrams]
A \textit{rooted} chord diagram is a chord diagram with a selected node. The selected node is called the \textit{root vertex}, and the chord with the root vertex is called the \textit{root chord}. In other words, a rooted chord diagram of size $n$ is a matching of the set $\{1,\ldots,2n\}$. For an algebraic definition, this is the same as a fixed-point free involution in $S_{2n}$. Then the generating series for rooted chord diagrams is 

\begin{equation}\label{rootedchorddiagsgen}
    D(x):=\sum_{n=0}(2n-1)!!\; x^n
\end{equation}

All chord diagrams considered here are going to be rooted and so, when we say a chord diagram we tacitly mean a rooted one.
\end{dfn}

Now, a rooted chord diagram can be represented in a linear order, by numbering the nodes in counterclockwise order, starting from the root which receives the label `$1$'. A chord in the diagram may be referred to as $c=\{a<b\}$, where $a$ and $b$ are the nodes in the linear order.

\begin{dfn}[Intervals]
In the linear form of a rooted chord diagram, an $interval$ is the space to the right of one of the nodes in the linear representation. Thus, a rooted diagram on $n$ chords has $2n$ intervals.
\end{dfn}
 
For example, this includes the space to the right of the last node (in linear order). 

\begin{figure}[!htb]
    \centering
    \includegraphics[scale=0.55]{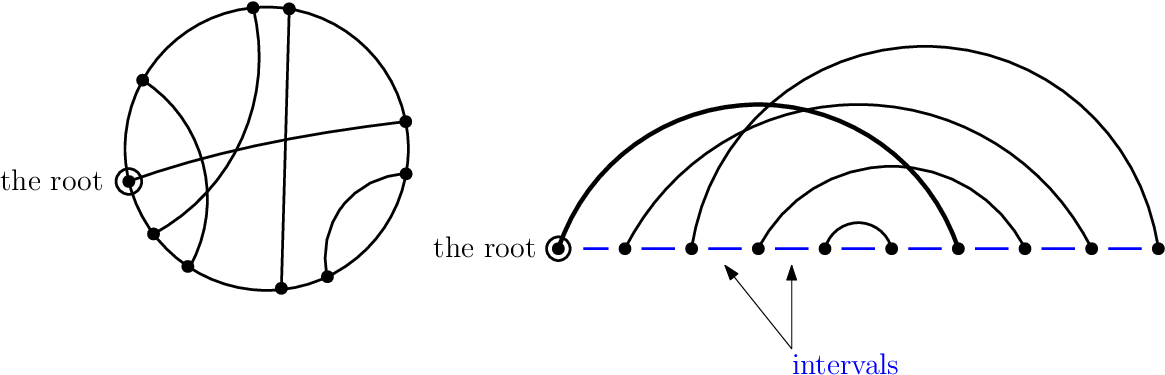}
    \caption{A rooted chord diagram and its linear representation}
\end{figure}

As may be expected by now, the crossings in a chord diagram encode much of the structure and so we ought to give proper notation for them. Namely, in the linear order, two chords $c_1=\{v_1<v_2\}$ and $c_2=\{w_1<w_2\}$ are said to \textit{cross} if $v_1<w_1<v_2<w_2$ or $w_1<v_1<w_2<v_2$.
Tracing all the crossings in the diagram leads to the following definition:

\begin{dfn}[The Intersection Graph] 
Given a (rooted) chord diagram $D$ on $n$ chords, consider the following graph $\mathcal{G}_D$: the chords of the diagram will serve as vertices for the new graph, and there is an edge between the two vertices $c_1=\{v_1<v_2\}$ and $c_2=\{w_1<w_2\}$ if $v_1<w_1<v_2<w_2$ or $w_1<v_1<w_2<v_2$, i.e. if the chords \textit{cross} each other. The graph so constructed is called the \textit{intersection graph} of the given chord diagram.  
\end{dfn}

\begin{rem}
A labelling for the intersection graph can be obtained as follows: give the label $1$ to the root chord;  order the components obtained if the root is removed according to the order of the first vertex of each of them in the linear representation, say the components are $C_1,\ldots,C_n$; and then recursively label each of the components. It is easily verified that a rooted chord diagram can be uniquely recovered from its labelled intersection graph.\\
\end{rem}

\begin{dfn}[Connected Chord Diagrams]\label{c}
A (rooted) chord diagram is said to be \textit{connected} if its  intersection graph is connected (in the graph-theoretic sense). A \textit{connected component} of a diagram is a subset of chords which itself forms a connected chord diagram. The term \textit{root component} will refer to the connected component containing the root chord.
\end{dfn}

\begin{exm}
The diagram $D$ below is a connected chord diagram in linear representation, where the root node is drawn in black.

\begin{center}
\includegraphics[scale=0.5]{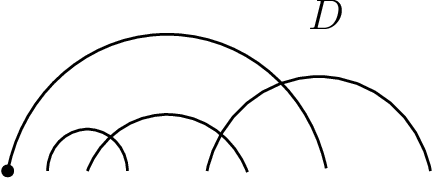}\end{center}
\end{exm}

The generating function for connected chord diagrams (in the number of chords) is denoted by $C(x)$. Thus $C(x)=\sum_{n=0}C_n x^n$, where $C_n$ is the number of connected chord diagrams on $n$ chords. The first terms of $C(x)$ are found to be 
\[C(x)=x+x^2+4x^3+27\;x^4+248\;x^5+\cdots\;;\]
the reader may refer to OEIS sequence \href{https://oeis.org/A000699}{A000699}  for more coefficients. 
The next lemma lists some classic decompositions for chord diagrams (see \cite{flajoletchords} for example).

\begin{lem}[\cite{flajoletchords, alipaper2con}]\label{cd}
  If $D(x), C(x)$ are the generating series for chord diagrams and connected chord diagrams respectively, then 
  \begin{enumerate}
      \item[$\mathrm{(i)}$]  $D(x)=1+C(xD(x)^2)$, \label{i1}
      \item[$\mathrm{(ii)}$] $D(x)=1+xD(x)+2x^2D'(x)$, and \label{i2}
      \item[$\mathrm{(iii)}$] $2xC(x)C'(x)=C(x)(1+C(x))-x$.\label{i3}
  \end{enumerate}

  \end{lem}

The main object we use throughout is chord diagrams with certain degrees (strengths) of connectivity.

\begin{dfn}[$k$-Connected Chord Diagrams]
A chord diagram on $n$ chords is said to be $k$-\textit{connected} if there is no set $S$ of consecutive endpoints, with $|S|< 2n-k$, $S$ is paired with less than $k$ endpoints not in $S$ (here we assume the endpoints are consecutive in the sense of the linear representation). In other words, the diagram requires the deletion of at least $k$ chords to become disconnected. A $k$-\textit{connected} diagram which is not $k+1$-\textit{connected} will be said to have \textit{connectivity} $k$.
\end{dfn}

The following relation between connected and 2-connected chord diagrams was established by the author in \cite{alipaper2con}.

\begin{prop}[\cite{alipaper2con}]\label{myproposition2connected}
The following functional relation between connected and $2$-connected diagrams holds:
\begin{equation}
    C=\displaystyle\frac{C^2}{x}-C_{\geq2}\left(\displaystyle\frac{C^2}{x}\right). \label{c2con}
\end{equation}
\end{prop}

The following corollary is about the asymptotic behaviour of $C_{\geq2}(x)$. See \cite{michi1} or Appendix A in \cite{alipaper2con} for more on factorially divergent power series and the terminology.

\begin{cor}[\cite{alipaper2con}]\label{imhere}
$C_{\geq2}(x)\in\mathbb{R}[[x]]^2_{\frac{1}{2}}$ and $
    \left(\mathcal{A}_{\frac{1}{2}}^2C_{\geq2}\right)(x)=
    \displaystyle\frac{1}{\sqrt{2\pi}}\cdot\displaystyle\frac{x^2}
   {C_{\geq2}S}\cdot  e^{\frac{-1}{2x}\left[\left(S+x\right)^2-1\right]},
$

where $S(x)=\displaystyle\frac{1}{\Big(1-\displaystyle\frac{C_{\geq2}(x)}{x}\Big)}$ is the generating series for sequences of $2$-connected chord diagrams counted by one less chord. Consequently, one computes that

\begin{align}
    (C_{\geq2})_n &= e^{-2} (2n-1)!!\bigg(1-\displaystyle\frac{6}{2n-1}-\displaystyle\frac{4}{(2n-3)(2n-1)}-\displaystyle\frac{218}{3}\displaystyle\frac{1}{(2n-5)(2n-3)(2n-1)}- \nonumber\\
    &\;   \qquad     -\displaystyle\frac{890}{(2n-7)(2n-5)(2n-3)(2n-1)}-\displaystyle\frac{196838}{15}\displaystyle\frac{1}{(2n-9)\cdots(2n-1)}-\cdots\bigg).\nonumber\\
    &\label{computC2asympt}
\end{align}
\end{cor}


\section{QED Theories, Quenched QED, and Yukawa Theory}

The two theories that we consider here are quenched QED and Yukawa theory, which are examples of QED-type theories. In these theories we have two particles, namely, fermion and boson (wiggly and dashed edges) particles, and we have only three-valent vertices of the type fermion-fermion-boson. We will compute the asymptotics of $\;z_{\phi_c|\psi_c|^2}(\hbar_{\text{ren}})\;$ in quenched QED, as well as the asymptotics of the green functions $\left.\partial^i_{\phi_c}(\partial_{\psi_c}\partial_{\bar{\psi}_c})^j\;G^{\text{Yuk}}\right|_{\phi_c=\psi_c=0}$. The approach is completely combinatorial and depends on establishing bijections between the diagrams in the combinatorial interpretation of the considered series and different classes of chord diagrams. Unlike the approach applied in \cite{michiq}, we do not need to refer to singularity analysis nor the representation of $\mathcal{S}(x)$ by affine hyperelliptic curves.

\subsection{The Partition Function}\label{partitionfunctionsection}
The partition function takes the form 
\begin{equation}
 Z(\hbar,j,\eta)=\int_\mathbb{R} \displaystyle\frac{1}{\sqrt{2\pi\hbar}}\; e^{\frac{1}{\hbar}\left(-\frac{x^2}{2}+jx+\frac{|\eta|^2}{1-x}+\hbar \log \frac{1}{1-x}\right)}dx .\label{Yu}\end{equation}

We are not going to discuss the physical reasoning behind the above expression, the reader may refer to QFT books or surveys for more details, e.g. see \cite{michi}. We only hint that, combinatorially, $\hbar \log \frac{1}{1-x}$ generates fermion loops, while $\frac{|\eta|^2}{1-x}$ generates a fermion propagator. The special examples of Yukawa theory and quenched QED will be as follows:

\begin{enumerate}
    \item Quenched QED is an approximation of QED where fermion loops are not present. So that the term $\hbar \log \frac{1}{1-x}$ does not appear in the partition function. Thus, the partition function for quenched QED is given by 

\[Z^{QQED}(\hbar,j,\eta)=\int_\mathbb{R} \displaystyle\frac{1}{\sqrt{2\pi\hbar}}\; e^{\frac{1}{\hbar}\left(-\frac{x^2}{2}+jx+\frac{|\eta|^2}{1-x}\right)}dx.\]
    
    \item For zero-dimensional Yukawa theory the partition function is just the integral in equation (\ref{Yu}). That is, the partition function for zero-dimensional Yukawa theory is given by 

\[Z^{Yuk}(\hbar,j,\eta)=\int_\mathbb{R} \displaystyle\frac{1}{\sqrt{2\pi\hbar}}\; e^{\frac{1}{\hbar}\left(-\frac{x^2}{2}+jx+\frac{|\eta|^2}{1-x}+\hbar \log\frac{1}{1-x}\right)}dx.\]
\end{enumerate}

In the next section we are going to display some of the literature on chord diagrams which will be essential for establishing the results.

\section{Quenched QED}

The case of quenched QED (QQED) was studied in \cite{alipaper2con}, so we will only state the results here for the sake of completeness. We studied the counterterm $\;z_{\phi_c|\psi_c|^2}(\hbar_{\text{ren}})\;$ obtained in \cite{michiq} (page 38). As mentioned earlier, QQED has a unique vertex type which is 3-valent. Hence, by \cite{michilattice}, the series $\;z_{\phi_c|\psi_c|^2}(\hbar_{\text{ren}})\;$ enumerates the number of primitive quenched QED diagrams with vertex-type residue (see sequence \href{https://oeis.org/A049464}{A049464} of the OEIS for the first entries).

This is the same as counting the number of all diagrams $\gamma$ with the following specifications:
\begin{enumerate}
    \item two types of edges, fermion and boson (photon) edges, represented as\; \raisebox{-0cm}{\includegraphics[scale=0.6]{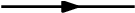}}\; and \raisebox{-0.23cm}{\includegraphics[scale=0.273]{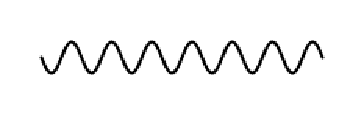}},  respectively;
    \item only three-valent vertices with the structure \raisebox{-0.64cm}{\includegraphics[scale=0.35]{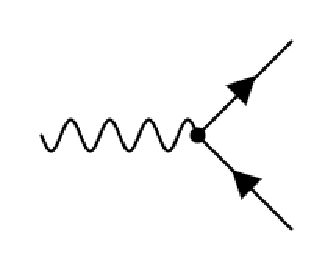}}, with one fermion in, one fermion out, and one photon;
    
    \item no fermion loops;
    \item the residue $\text{res} (\gamma)$ is vertex-type; and
    \item $\gamma$ is $1PI$ primitive, in other words it is edge-connected and contains no subdivergences.
\end{enumerate}

\begin{thm}[\cite{alipaper2con}]\label{myresultinquenched}
The generating series for the renormalization counterterms $\;z_{\phi_c|\psi_c|^2}(\hbar_{\text{ren}})\;$ and $\;z_{|\psi_c|^2}(\hbar_{\text{ren}})\;$ count $2$-connected chord diagrams. More precisely, $$[\hbar_{\text{ren}}^{n-1}]\;z_{\phi_c|\psi_c|^2}(\hbar_{\text{ren}})=
[\hbar_{\text{ren}}^{n}]\;z_{|\psi_c|^2}(\hbar_{\text{ren}})=[x^n]\;C_{\geq2}(x).$$
\end{thm}

This says that we can fathom the asymptotic behaviour of these counterterms by our knowledge of 2-connected chord diagrams, represented by Corollary \ref{imhere}. This gives 

\begin{align}
    [\hbar_{\text{ren}}^{n-1}]\;&z_{\phi_c|\psi_c|^2}(\hbar_{\text{ren}})=
[\hbar_{\text{ren}}^{n}]\;z_{|\psi_c|^2}(\hbar_{\text{ren}})=[x^n]\;C_{\geq2}(x)=\nonumber\\
    & = e^{-2} (2n-1)!!\bigg(1-\displaystyle\frac{6}{2n-1}-\displaystyle\frac{4}{(2n-3)(2n-1)}-\displaystyle\frac{218}{3}\displaystyle\frac{1}{(2n-5)(2n-3)(2n-1)}- \nonumber\\
    &\;   \qquad     -\displaystyle\frac{890}{(2n-7)(2n-5)(2n-3)(2n-1)}-\displaystyle\frac{196838}{15}\displaystyle\frac{1}{(2n-9)\cdots(2n-1)}-\cdots\bigg),
\end{align} which matches with the calculation through singularity analysis made in \cite{michiq}.


\section{Yukawa Theory}\label{YYYYY}
In this section we will establish the connection between interacting Yukawa theory and chord diagrams. Unlike the case of quenched QED, the relation this time is well hidden. The problem was based on an observation, made by the author, upon seeing Table 18 (a) in \cite{michiq} while working on the quenched QED case. For the sake of clarity, let us display Table 18 (a) of \cite{michiq} (with an extra column added to emphasize the relation to chord diagrams):

\begin{table}[h]
\center
\begin{tabular}{|c|c||c|c|c|c|c|c|c|}\hline
 &            &$\hbar^0$ & $\hbar^1$   & $\hbar^2$  & $\hbar^3$ & $\hbar^4$ & $\hbar^5$ & $\hbar^6$                           \\\hline\hline
1&$\partial_{\phi_c}^0(\partial_{\psi_c}\partial_{\bar{\psi}_c})^0
G^{\text{Yuk}}\big|_{\phi_c=\psi_c=0}$     
              &   0  & 0     & 1/2    & 1     & 9/2   & 31    & 283   
              \\\hline
2&$\partial_{\phi_c}^1(\partial_{\psi_c}\partial_{\bar{\psi}_c})^0
G^{\text{Yuk}}\big|_{\phi_c=\psi_c=0}$      
              &   0  & 1     & 1      & 4     & 27    & 248   & 2830       
              \\\hline
3&$\partial_{\phi_c}^2(\partial_{\psi_c}\partial_{\bar{\psi}_c})^0
G^{\text{Yuk}}\big|_{\phi_c=\psi_c=0}$      
              &  -1  & 1     & 3      & 20    & 189   & 2232  & 31130      
              \\\hline
4&$\partial_{\phi_c}^0(\partial_{\psi_c}\partial_{\bar{\psi}_c})^1
G^{\text{Yuk}}\big|_{\phi_c=\psi_c=0}$ 
              &  -1  & 1     & 3      & 20    & 189   & 2232  & 31130   
              \\\hline
5&$\partial_{\phi_c}^1(\partial_{\psi_c}\partial_{\bar{\psi}_c})^1
G^{\text{Yuk}}\big|_{\phi_c=\psi_c=0}$
              &   1  & 1     & 9      & 100   & 1323  & 20088 & 342430       
              \\\hline
 \end{tabular}\vspace{0.3cm}\caption{The first coefficients of the proper Green functions $\left.\partial_{\phi_c}^i(\partial_{\psi_c}\partial_{\bar{\psi}_c})^jG^{\text{Yuk}}\right|_{\phi_c=\psi_c=0}$  of Yukawa theory such that $i+2j\in\{0,1,2,3\}$.  }
\label{table4}
\end{table}
As we have mentioned in Section \ref{partitionfunctionsection}, Yukawa theory is different from quenched QED in that  fermion loops are allowed in the diagrams. Thus, if $\mathcal{U}_{ij}$ denotes the class of 1PI Feynman graphs counted by the Green function $\left.\partial_{\phi_c}^i(\partial_{\psi_c}\partial_{\bar{\psi}_c})^jG^{\text{Yuk}}\right|_{\phi_c=\psi_c=0}$, then, to our combinatorial concern,   $\mathcal{U}_{ij}$ consists of all graphs $\gamma$ with the following specifications:
\begin{enumerate}
   \item two types of edges (as before), fermion and boson (meson) edges, represented as \raisebox{-0cm}{\includegraphics[scale=0.6]{Figures/fermionedge.eps}} \;and \raisebox{-0.2cm}{\includegraphics[scale=0.4]{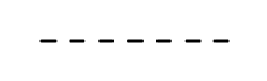}},  respectively;
    \item only three-valent vertices with the structure \raisebox{-0.64cm}{\includegraphics[scale=0.3]{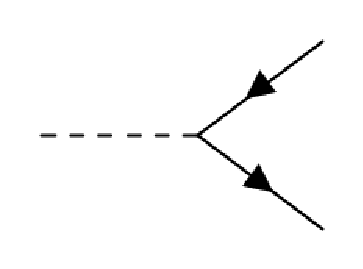}}, with one fermion in, one fermion out, and one boson;
    
    \item fermion loops are allowed;
    \item the residue $\text{res} (\gamma)$ has $i$ external boson legs,  $j$  external fermion-in legs, and $j$ external fermion-out legs; and 
    \item $\gamma$ is $1PI$, i.e. it has no bridges (2-edge-connected). This is implied by the definition of proper Green functions.
\end{enumerate}

Note that, unlike the quenched QED case, we are not restricting to primitive diagrams since we are not working with any expressions from renormalization in this part (yet).

\begin{lem}\label{property1Yukawa}
Let $\Gamma$ be a Yukawa theory 1PI graph with $2j$ external fermion legs. Then
\begin{equation}
    |V(\Gamma)|=f+j,
\end{equation}

where, as before, $|V(\Gamma)|$ is the number of vertices and $f$ is the number of internal fermion edges.
\end{lem}

\begin{proof}
Consider the graph $\Gamma'$ obtained from $\Gamma$ by removing all boson edges and half edges. The resulting graph $\Gamma'$ is generally a collection of fermion loops and fermion paths. 
The number of these paths should be $j$. Indeed, on one hand every such path will give two external fermion legs when the boson edges are present. To see this notice that, under the vertex condition, such paths can not end with a vertex: the boson edges can not then complete the degree of such a vertex. On the other hand, it is clear that the external fermion legs can only be at the ends of such paths. 

As a consequence of the above argument, the number of vertices can be calculated as follows: every fermion loop has  as many vertices as fermion edges. Whereas in every fermion path the number of vertices is more by one the number of internal fermion edges. This proves the lemma.
\end{proof}

Next we investigate the combinatorial meaning of the Green functions in 
Table \ref{table4} and its relation to chord diagrams. 

\subsection{Yukawa Tadpole Graphs:\quad $\left.\partial_{\phi_c}^1(\partial_{\psi_c}\partial_{\bar{\psi}_c})^0G^{\text{Yuk}}(\hbar,\phi_c,\psi_c)\right|_{\phi_c=\psi_c=0}$}

By the definition mentioned earlier, $\left.\partial_{\phi_c}^1(\partial_{\psi_c}\partial_{\bar{\psi}_c})^0G^{\text{Yuk}}(\hbar,\phi_c,\psi_c)\right|_{\phi_c=\psi_c=0}$ is the generating series of Yukawa theory graphs with exactly one external leg, which is of boson type, graded by loop number. In other words, 
\[[\hbar^n]\left.\partial_{\phi_c}^1(\partial_{\psi_c}\partial_{\bar{\psi}_c})^0G^{\text{Yuk}}\right|_{\phi_c=\psi_c=0}\]
is the number of \textit{1PI} tadpole graphs with one boson leg and loop number $n$.

From Table \ref{table4}, we can conjecture that $[\hbar^n]\left.\partial_{\phi_c}^1(\partial_{\psi_c}\partial_{\bar{\psi}_c})^0G^{\text{Yuk}}\right|_{\phi_c=\psi_c=0}=C_n$, the number of connected chord diagrams on $n$ chords. Interestingly, unlike the quenched QED graphs, tadpoles do a great job hiding their chord diagrammatic structure. This however is to be unveiled Theorem \ref{myresult1inYukawa} below. In Figure \ref{27tadpoles} below, we display the tadpoles counted in $[\hbar^4]\left.\partial_{\phi_c}^1(\partial_{\psi_c}\partial_{\bar{\psi}_c})^0G^{\text{Yuk}}\right|_{\phi_c=\psi_c=0}$, which, by our claim, should be $27$ as $C_4$. 

\begin{rem} For the sake of simplicity of drawings, we will drop the direction of the fermion loops and assume it is always counter-clockwise. Besides, we will draw no more dashed boson lines, and shall instead use a light line for bosons and a heavier line for fermions. Note that the relative direction of loops matters in some cases and have to be compensated sometimes by a twist in the boson edges. For example,  the following two tadpoles in Figure \ref{two diff tadpoles1} are different, and shall be represented as in Figure \ref{two diff tadpoles11}:
\begin{figure}[h]
    \centering
   \raisebox{0cm}{\includegraphics[scale=0.35]{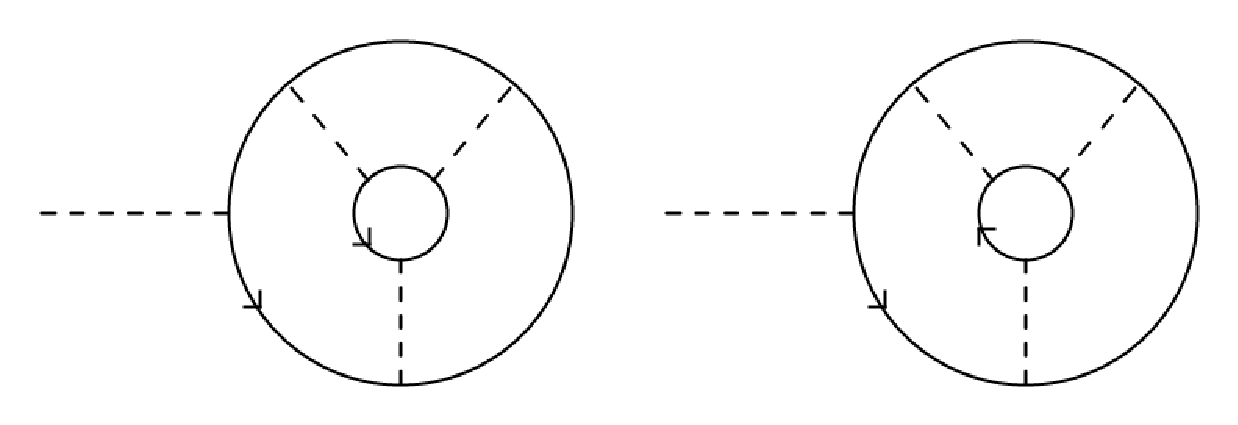}}
    \caption{Two tadpoles may differ due to the relative orientation of fermion loops.}
    \label{two diff tadpoles1}
\end{figure}
\begin{figure}[h]
    \centering
   \raisebox{0cm}{\includegraphics[scale=0.45]{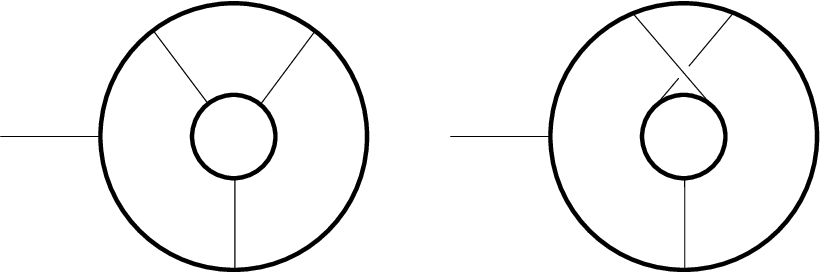}}
    \caption{All loops are now assumed oriented counter-clockwise and the attached boson edges have to be twisted accordingly; Fermion edges are drawn thicker than boson edges.}
    \label{two diff tadpoles11}
\end{figure}
\end{rem}

\begin{figure}
    \centering
   \raisebox{0cm}{\includegraphics[scale=0.8
   ]{
   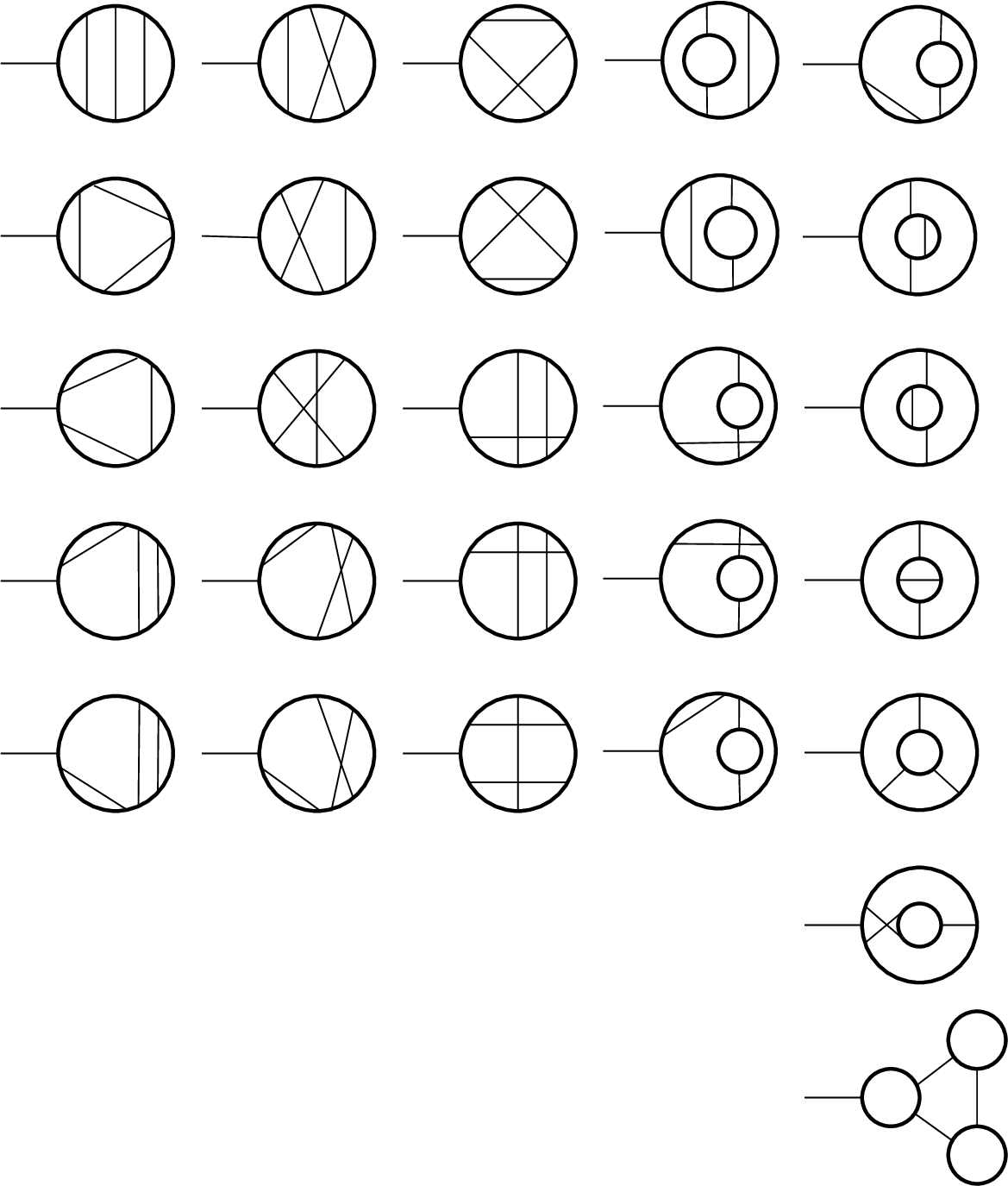}}
    \caption{The 27 \textit{1PI } tadpole graphs with loop number 4}
    \label{27tadpoles}
\end{figure}

The next lemma is a direct corollary to Lemma \ref{property1Yukawa} for the case of tadpoles.
\begin{lem}\label{property2yukawa}
For any Yukawa 1PI tadpole graph $\Gamma\in\mathcal{U}_{10}$ the following are true:

\begin{enumerate}
   
\item The number of vertices is equal to the number of fermion edges,
\[|V(\Gamma)|=f.\]

 \item The number of all boson edges (including the external boson leg) is equal to the loop number of the graph. That is,
\[p+1=\ell(\Gamma).\]

\item Fermion loops partition the set of vertices;
\end{enumerate}
 where, as before, $p$ is the number of internal boson edges (the $p$ is suitable in this case too as it stands for a \textit{pion}, pions  are often the bosons in a Yukawa interaction), and $f$ is the number of fermion edges (all fermion edges are internal in this case).

\end{lem}

\begin{proof}
By Lemma \ref{property1Yukawa}, we directly have $|V(\Gamma)|=f$.  Euler's formula now implies   \[1=|V(\Gamma)|-(p+f)+\ell(\Gamma)=f-(p+f)+\ell(\Gamma)=-p+\ell(\Gamma),\] for any $\Gamma\in\mathcal{U}_{10}$. Finally, since there are no external fermion edges, every fermion edge must be on a fermion loop. Thus all vertices are on fermion loops, and by the condition on the vertices in the theory, no vertex can lie on more than one fermion loop.
\end{proof}

Lemma \ref{property2yukawa} is useful in that we do not have to think about the loop number in proving the bijection to connected chord diagrams, and can instead focus on the more evident count of boson edges.

\begin{nota}\label{nota randv}
For a tadpole $T\in\mathcal{U}_{10}$, we will fix the notation that the external boson leg is denoted by $r_T$, and the vertex at the leg is denoted by $v_T$ (this is consistent with the notation used in the previous section). For a vertex $a \in V(T)$ we let Loop$(a)$ denote the unique fermion loop containing $a$, and we let Fermion$(a)$ be the fermion edge coming out of $a$ (i.e. the next on Loop$(a)$ counter-clockwise). Boson$(a)$ will denote the unique boson edge to which $a$ is incident.
\begin{center}

\raisebox{0cm}{\includegraphics[scale=0.5]{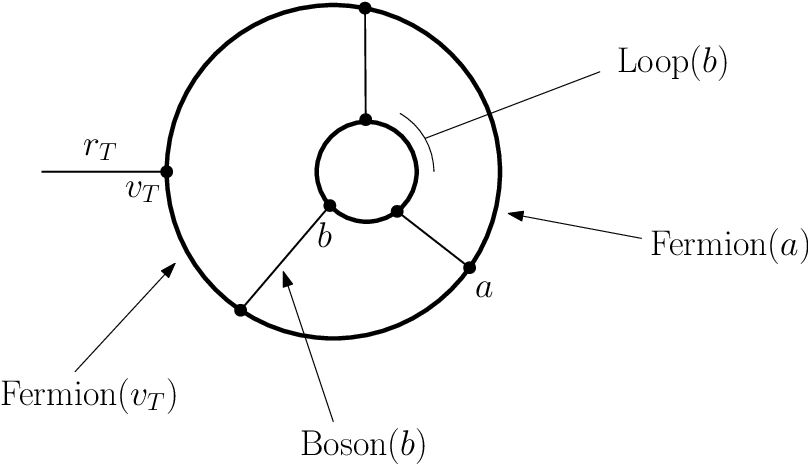}}
    
\end{center}

In the next proof we consider the free end of $r$ to be a vertex of degree 1. Then the number of vertices is twice the number of boson edges. 
\end{nota}

\begin{thm}\label{myresult1inYukawa}
The number of Yukawa 1PI tadpole graphs with loop number $n$ is equal to the number of connected chord diagrams on $n$ chords. In other words,
\[[\hbar^n]\left.\partial_{\phi_c}^1(\partial_{\psi_c}\partial_{\bar{\psi}_c})^0G^{\text{Yuk}}\right|_{\phi_c=\psi_c=0}=C_n.\]

\end{thm}

\begin{proof}
Let $T(x)$ be the generating series for tadpoles in $\mathcal{U}_{10}$, counted by the number of boson edges (including the external boson leg). We are taking advantage of Lemma \ref{property2yukawa} in order to use the number of bosons instead of the loop number. The theorem shall be proven through an algorithm that shows that  $T(x)$ obeys the same recurrence as $C(x)$ (see Lemma \ref{cd}), namely
\[2xT(x)T^{\prime}(x)=T(x)^2+T(x)-x.\]
First notice that the LHS stands for two tadpole diagrams, one of which has a distinguished end point of one of the boson edges (hence the 2 factor). For simplicity, we will treat the free end of an external boson leg as a vertex. Then let $\mathcal{U}_{10}^\bullet$ be the class of tadpoles with a distinguished vertex. Also let $\mathcal{U}_{10}-\{\mathcal{X}\}$ be the class of tadpoles excluding $\mathcal{X}$, the tadpole with one vertex ($\mathcal{X}$ has only the boson leg).

Let $T_1,T_2$ be tadpole graphs in $\mathcal{U}_{10}$, and assume that $T_2$ has a distinguished vertex $d$. By Notation \ref{nota randv}, we let $r_1$ and $r_2$ be the boson legs of  $T_1$ and $T_2$, and we let $v_1$ and $v_2$ be the 3-valent vertices incident to $r_1$ and $r_2$, respectively. Figure \ref{T1T2figure} illustrates the notation.

\begin{figure}[h!]
    \centering
    \raisebox{0cm}{\includegraphics[scale=0.5]{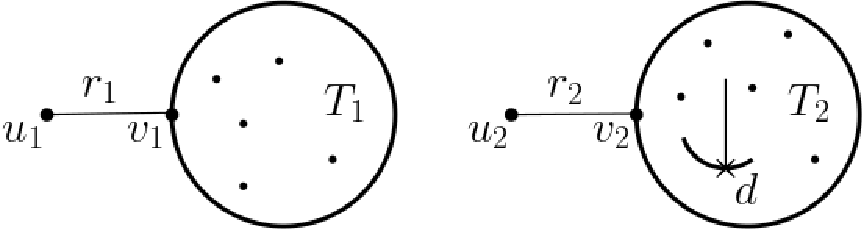}}
    \caption{Notation for $(T_1,(T_2,d))$}
    \label{T1T2figure}
\end{figure}

\setlength{\parindent}{0cm}
Now we can describe the reversible algorithm as follows:

\rule{\textwidth}{0.4pt}
\textbf{Algorithm $\Psi$:} \quad$(\mathcal{U}_{10}\times\mathcal{U}_{10}^\bullet)\longrightarrow (\mathcal{U}_{10}\times\mathcal{U}_{10})\;\bigcup\; (\mathcal{U}_{10}-\{\mathcal{X}\})$\\
\rule{\textwidth}{0.4pt}

\textbf{Input:} $(T_1,(T_2,d)) \in(\mathcal{U}_{10}\times\mathcal{U}_{10}^\bullet)$, with notation as described above.\\

 (a) If $d=u_2$ just \textbf{return} $(T_1,T_2)$.\\
 
 (b) If $d\neq u_2$, do the following: 
 
 Move (counter-clockwise) along Loop$(v_1)$ in $T_1$, determine Fermion$(v_1)$ and let $w$ be the first vertex met on the loop. Note that $w$ may be $v_1$ itself.
 
    \begin{enumerate}
  \item  If $w=v_1$, i.e. $T_1$ contains no internal boson edges, \textbf{return} the tadpole $T$ obtained as follows:
  
  \begin{enumerate}
      \item[(i)] Insert vertex $v_1$ together with the leg $r_1$ into Fermion$(d)$ in $T_2$ by making a subdivision of Fermion$(d)$.
      \item[(ii)] Insert $u_2$ into the new Fermion$(v_1)$ on Loop$(d)$.
  \end{enumerate}
  
  \item If $w\neq v_1$, \textbf{return} the tadpole $T$ obtained as follows:
  \begin{enumerate}
      \item[(i)] Insert $u_2$ into  Fermion$(v_1)$ in $T_1$. 
      \item[(ii)] Detach $w$ from Loop$(v_1)$ and insert it into Fermion$(d)$ in $T_2$.
  \end{enumerate}
 
        \end{enumerate}
     \rule{\textwidth}{0.4pt}       
     
Figures \ref{Psi1} and \ref{Psi2} illustrate the two cases for (b).

\begin{figure}[h!]
    \centering
    \includegraphics[scale=0.6]{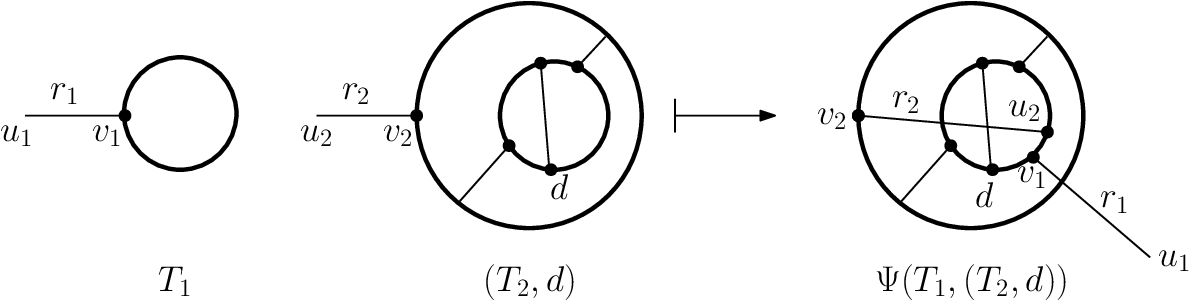}
    \caption{$\Psi(T_1,(T_2,d))$ when $T_1$ has exactly one vertex, the case $w=v_1$.}
    \label{Psi1}
\end{figure}

    \begin{figure}[h!]
    \centering
    \includegraphics[scale=0.56]{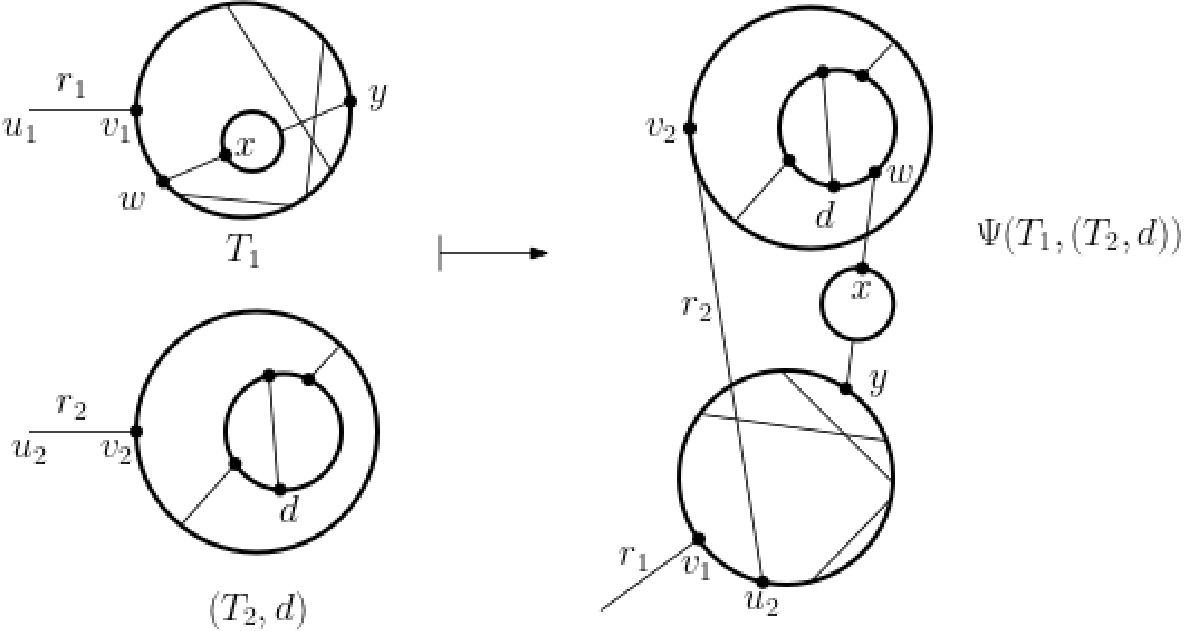}
    \caption{$\Psi(T_1,(T_2,d))$ for a general $T_1$, the case $w\neq v_1$.}
    \label{Psi2}
\end{figure}

For the reverse process we can devise the following algorithm.

\rule{\textwidth}{0.4pt}
\textbf{Algorithm $\Psi^{-1}$:} \quad$ (\mathcal{U}_{10}\times\mathcal{U}_{10})\;\bigcup\; (\mathcal{U}_{10}-\{\mathcal{X}\})\longrightarrow(\mathcal{U}_{10}\times\mathcal{U}_{10}^\bullet)$\\
\rule{\textwidth}{0.4pt}

\textbf{Input:} $(T_1,T_2)\in(\mathcal{U}_{10}\times\mathcal{U}_{10})$,\; or \;$ T \in \mathcal{U}_{10}-\{\mathcal{X}\}$\\

 (a) If the input is a pair $(T_1,T_2)$ then \textbf{return} $(T_1,(T_2,v_2))$.\\
 
 (b) If the input is a tadpole graph $T\in\mathcal{U}_{10}-\{\mathcal{X}\}$, do the following: 
 \begin{enumerate}
 \item Move (counter-clockwise) along Loop$(v_T)$, determine the first vertex $a$ met on the loop. Note that $a\neq v_T$  since we are excluding the tadpole with a single vertex.

  \item  Determine the other end vertex of Boson$(a)$ and denote it by $v_2$.
  
  \item Remove the vertex $a$ from $T$ and keep the resulting boson leg attached at $v_2$. Let the resulting graph be denoted by $\Gamma$.
  
  \item Check whether $\Gamma$ contains a bridge:
  
  \begin{enumerate}
      \item[Case 1:] If $\Gamma$ is 2-edge-connected (i.e. contains no bridges) do
      \begin{enumerate}
          \item[(1)] Determine the first vertex on Loop$(v_T)$ before $v_T$, denote it by $d$.
          \item[(2)] Remove $v_T$ and its boson leg $r_T$. Denote the remaining tadpole by $T_2$.
          
          \item[(3)] \textbf{Return} $(\mathcal{X}, (T_2,d))$.
      \end{enumerate}

\item[Case 2:] If $b_0$ is a bridge (must be a boson edge) in $\Gamma$,  undergo the following:
      \begin{enumerate}
        
    \item[(1)] Set $G=\Gamma$ and $b=b_0$. 
    \item[(2)] \{\textbf{while} $G$ has a bridge $g$ \textbf{do}
      \begin{enumerate}
      \item reset $b\longleftarrow g$;
      \item determine the component $\gamma$ that contains $v_2$ if $b$ is removed; 
      \item reset $G\longleftarrow \gamma$.\}
      \end{enumerate}
   \item[(3)] Let $w$ be the end vertex of $b$ that lies in $G$. Notice that, after the while-loop, $G$ contains no bridges. 
    \item[(4)] Determine the first vertex on Loop$(w)$ before $w$, denote it by $d$.
    \item[(5)] Detach $w$ from Loop$(d)$ in $G$ and insert into Fermion$(v_T)$ (i.e. next to $v_T$ on Loop$(v_T)$).
    \item[(6)] Let $T_2=G-{w}$ be the tadpole obtained from $G$ after $w$ is removed.
    \item[(7)] Let $T_1$ be that tadpole obtained in $\Gamma-G$ after $w$ is inserted on Loop$(v_T)$.
    \item[(8)] \textbf{Return} $(T_1,(T_2,d))$.
      \end{enumerate}
  \end{enumerate}
  \end{enumerate}
 \rule{\textwidth}{0.4pt}
Before discussing the algorithms, the reader may like to see Example \ref{examplePsiinverse} for applying $\Psi^{-1} $ to $\Psi(T_1,(T_2,d))$ from Figure \ref{Psi2}.

Now, for Algorithm $\Psi$, the two cases (a) and (b) are clearly distinguishable by the types of their outputs, so Let us focus on (b). 

\begin{enumerate}
    \item The special case in ($\Psi$ b:(1)) when $w=v_1$ returns a tadpole without bridges. Indeed, we only added an external leg $r_1$ at the position determined by $d$, and then we inserted the free end $u_2$ (the external leg of $T_2$). None of these steps changes the connectivity of $T_2$, and the return value is indeed in $\mathcal{U}_{10}-\{\mathcal{X}\}$. 
    
    \item In ($\Psi$ b:(2)), when $w\neq v_1$, the result has no bridges, since what we do is roughly joining $T_1$ and $T_2$ by means of two boson edges in a certain way. Thus the result is indeed a tadpole in $\mathcal{U}_{10}-\{\mathcal{X}\}$. Notice that the  first of these joints is attached next to the leg, and its removal leaves the graph with a bridge. This is of most importance in the reverse process.
 
\end{enumerate}
 
 Then, for Algorithm $\Psi^{-1}$, we have the following:
 
 \begin{enumerate}
     \item If the input is a pair, then this uniquely means that the distinguished vertex satisfies $d=u_2$.
     \item If the input is a tadpole that stays bridgeless after the the vertex $a$ next to $v_T$ is  removed then this uniquely means that $T_1=\mathcal{X}$. Indeed, we have seen above that in all the other cases we get a bridge if the first boson edge after the external leg is removed. 
     
     \item If the input reveals a bridge $b_0$ when Boson$(a)$ is removed, then we learn that $a$ and Boson$(a)$ formed the external leg of $T_2$ and we start disentangling $T_2$ from the graph.   Roughly speaking, we need to determine $d$ by using the fact that, in the absence of Boson$(a)$, Boson$(d)$ is a bridge coming from $T_1$. The while loop in the algorithm works on finding the last such bridge.
     
     \item By the engineering of the while-loop, the graph $G$ obtained at the end of the loop has no more bridges, besides, it carries the traces of the last bridge $b$, which determines the distinguished vertex $d$. 
     
     \item After modifying $G$ by removing $w$ we get $(T_2,d)$, and simultaneously we get $T_1$ from the remaining graph by attaching $w$ into Fermion$(v_T)$. By doing so, $T_1$ is also bridgeless.
 \end{enumerate}
This proves the theorm.
\end{proof}
\begin{exm}\label{examplePsiinverse}
Let us apply Algorithm $\Psi^{-1}$ to the tadpole $T$ given in Figure \ref{Psi2} by 
\begin{center}
    \raisebox{0cm}{\includegraphics[scale=0.4]{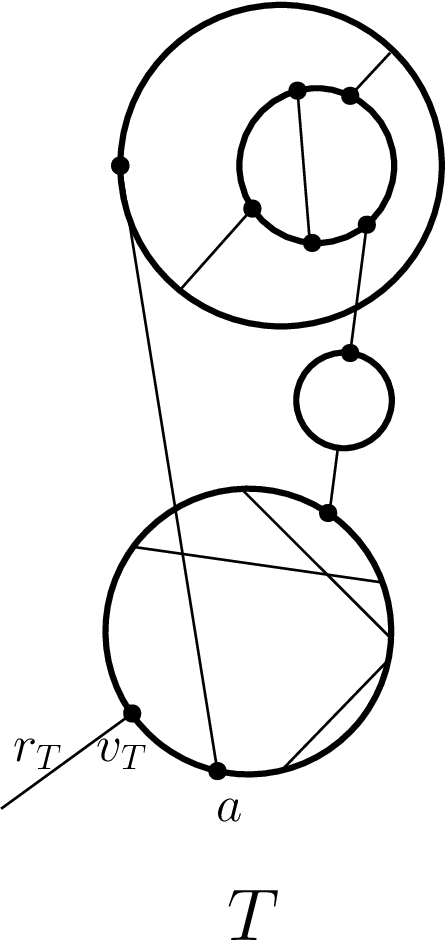}}
\end{center}

\begin{enumerate}
    \item The input is a tadpole in $\mathcal{U}_{10}-\{\mathcal{X}\}$ and so we apply (b).
    \item We determine $a$ as the vertex next to $v_T$ on Loop$(v_T)$, and with it we determine $v_2$, the other end of Boson$(a)$.
    \item We remove vertex $a$ from $T$ and keep Boson$(a)$ attached at $v_2$ as in the figure below.
    
    \begin{center}
    \raisebox{0cm}{\includegraphics[scale=0.56]{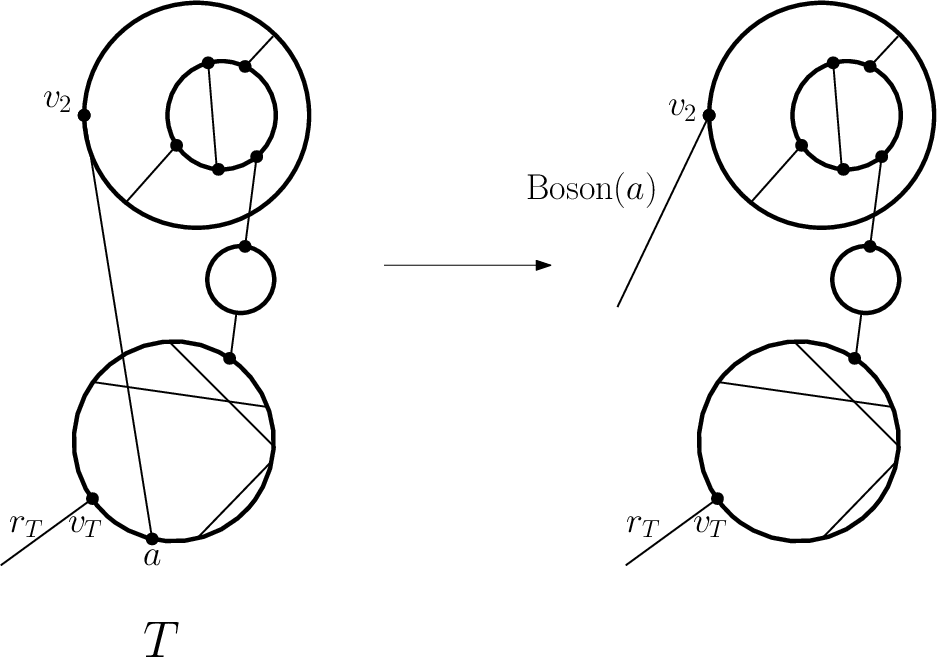}}
\end{center}

\item We check for bridges and we find one of them, we assume that our search provided the bridge $b_0$.
\begin{center}
    \raisebox{0cm}{\includegraphics[scale=0.45]{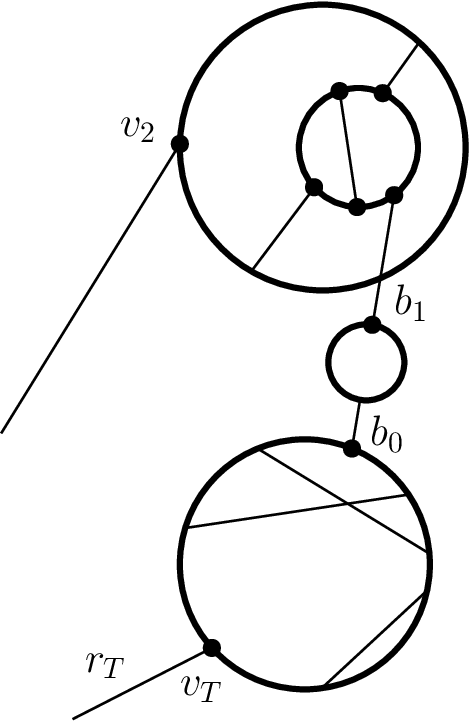}}
\end{center}

\item We enter the while loop with $G=\Gamma$ given above and $b=b_0$:\begin{enumerate}
    \item After the first iteration $G$ is modified to be 
    \begin{center}
    \raisebox{0cm}{\includegraphics[scale=0.53]{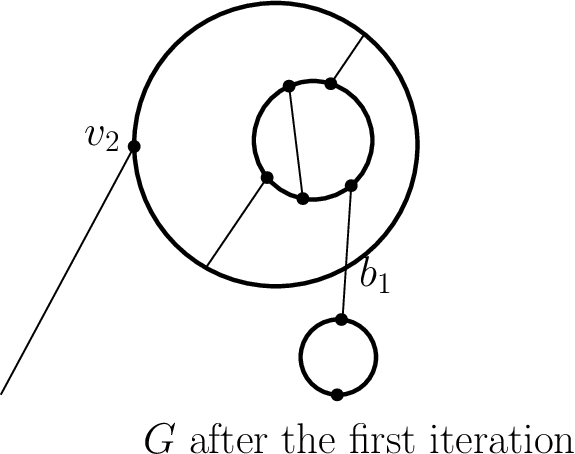}}\end{center}
    \vspace{0.3cm}
    \item Now $G$ again has a bridge $b$, and, after the second iteration $G$ has no more bridges. After the while-loop $b=b_1$ and $G$ is given by

\begin{center}
    \raisebox{0cm}{\includegraphics[scale=0.53]{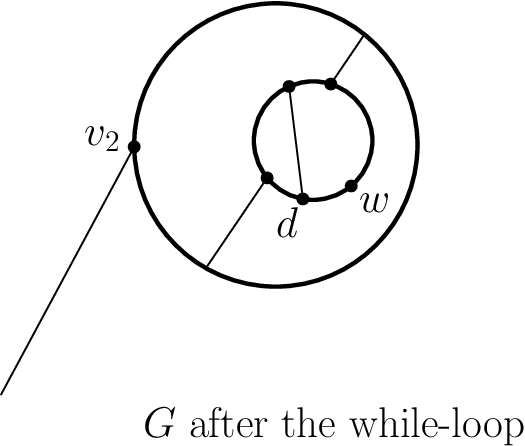}}\end{center}
    \vspace{0.3cm}
    
    \item Finally, detach $w$ from $G$, set $T_2=G-w$, reset $G=G-w$, and insert $w$ next to $v_T$ in $\Gamma-G$ to get $T_1$. The result is shown in the figure below.
\end{enumerate} 

    \raisebox{0cm}{\includegraphics[scale=0.55]{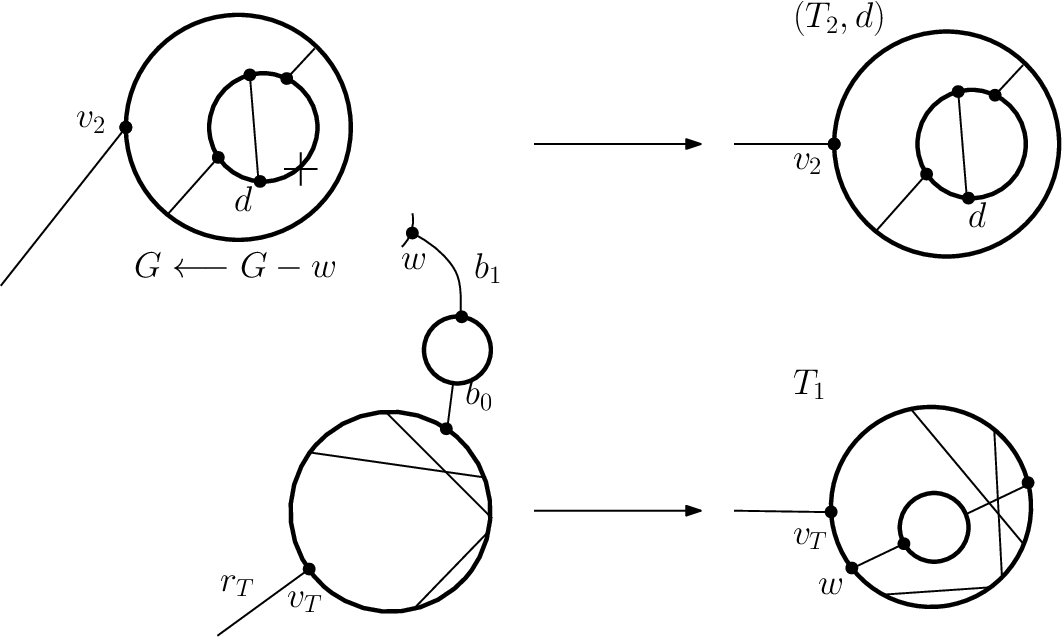}}
\end{enumerate}

\end{exm}

Theorem \ref{myresult1inYukawa} tells us that the two structures, Yukawa 1PI tadpoles and connected chord diagrams satisfy the same recurrence and hence there exists a bijection between the two classes obtained recursively. However, we still have to do one more bit of work to express this bijection. The bijection should respect the sizes, that is, a Yukawa 1PI tadpole with $n$ boson edges (with $n$ loops) should be mapped to a connected chord diagram on $n$ chords.

Moreover, as we can see, the vertices of a tadpole should correspond to the vertices in a connected chord diagram, and the fermion edges should accordingly correspond to the intervals. It has not been made clear so far  how we can order fermion edges in a way that resembles the natural linear order of the intervals in chord diagrams, an  order that is compatible with the decomposition in Theorem \ref{myresult1inYukawa}. Definition \ref{psiorder} below addresses this issue. To see why the order should be defined this way, we have to first recall the root share decomposition of connected chord diagrams. The root share decomposition has been used in the proof of Lemma \ref{cd}, and now we need to define it properly:

\begin{dfn}[Root Share Decomposition]
The \textit{root share decomposition} is the map $\nabla: \mathcal{C}\longrightarrow \mathcal{C}\times(\mathcal{C},\mathbb{N})$ defined by 

\[\nabla C=(C_1,(C_2,k)),\]

where $1\leq k\leq |C_2|-1$, and $C_1,C_2$ are obtained as follows:
Among the components produced by removing the root of $C$, $C_2$ is taken to be the first in intersection order with the root. $k$ determines the interval in $C_2$ through which the root used to pass. $C_1$ is then obtained by removing the chords of $C_2$ from $C$. 

If $(C_1,(C_2,k))\in\mathcal{C}\times(\mathcal{C},\mathbb{N})$ is a valid triplet then $\nabla^{-1}(C_1,(C_2,k))$ is the connected chord diagram obtained by placing $C_1$ in the $k$th interval of $C_2$ and pulling the root of $C_1$ out to become the root of the whole diagram (i.e. place it to the left of the root of $C_2$). See Figure \ref{roootshare}.
\end{dfn}

\begin{figure}[h!]
    \centering
    \includegraphics[scale=0.86]{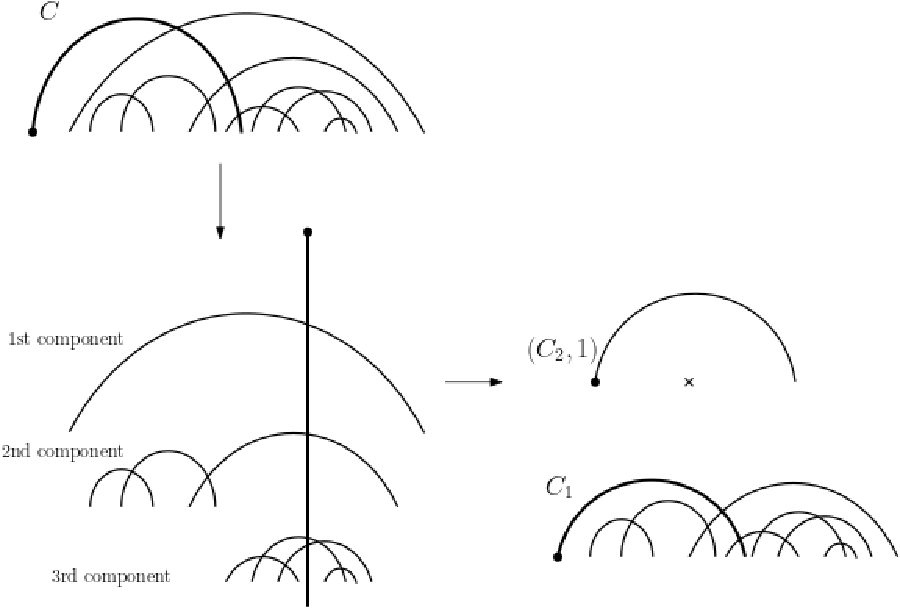}
    \caption{Root share decomposition of a connected chord diagram}
    \label{roootshare}
\end{figure}

The effect of the root share decomposition on the linear order of the intervals in $C_1$ and $C_2$ leads us to the following order on fermion edges in a Yukawa 1PI tadpole.

\begin{dfn}[The $\Psi$-order]\label{psiorder}
We define the \textit{$\Psi$-order} on the fermion edges in a Yukawa 1PI tadpole inductively on the size $n$ of the tadpoles as follows:

 $\ast$ For a fermion edge $e$ in a tadpole $T$, its  $\Psi$-order takes values in $\mathbb{N}$ and is to be denoted by $\psi_T(e)$.
\begin{itemize}
   
    \item For $\mathcal{X}$ the unique fermion edge is ordered as 1.
    \item Assume all tadpoles of size less than $n$ are ordered and let $T$ be a tadpole of size $n$. To order $T$ do the following:
    \begin{enumerate}
        \item Apply $\Psi^{-1}$ to $T$ to determine a triplet $(T_1,(T_2,d))$. As before, let $v_T$, $v_1$ and $v_2$ be the leg vertices in $T$, $T_1$ and $T_2$ respectively. Let $w_1$ be the vertex next to $v_1$ in $T_1$. Also let $w_T$ be the vertex in $T$ next to $v_T$, and let $w_d$ be the vertex in $T$ next to the vertex $d$ (i.e. these are the vertices of subdivisions created by $\Psi$). Note that $w_d=v_T$ if $T_1=\mathcal{X}$.
        
        \item By the induction hypothesis it is assumed that we know the $\Psi$-ordering of $T_1$ and $T_2$. Let \[M=\max_{e\in T_1} \psi_{T_1}(e).\]
        
        \end{enumerate}
        \begin{enumerate}
        \item[Case 1:]  $T_1=\mathcal{X}$. Set $\psi_T(\text{Fermion}(v_T))=1$, and for any other fermion edge $e$ in $T$ define $\psi_T(e)$ such that
        
        \begin{enumerate}
            \item $\psi_T(e)=\psi_{T_2}(e)+1$ if $ e\in T_2$ and $\psi_{T_2}(e)< \psi_{T_2}(\text{Fermion}(d))$.
            
            \item $\psi_T(\text{Fermion}(d))=\psi_{T_2}(\text{Fermion}(d))+1$ \quad (in $T$ this is the $dv_T$ edge).

            \item $\psi_T(\text{Fermion}(w_T))=\psi_{T_2}(\text{Fermion}(d))+2$.
            
            \item $\psi_T(e)=\psi_{T_2}(e)+2$ if $e\in T_2$.
            \end{enumerate}
        
        \item[Case 2:] $T_1\neq\mathcal{X}$. Set $\psi_T(\text{Fermion}(v_T))=1$, and for any other fermion edge $e$ in $T$ define $\psi_T(e)$ such that
        
        \begin{enumerate}
            \item $\psi_T(e)=\psi_{T_2}(e)+1$ if $ e\in T_2$ and $\psi_{T_2}(e)< \psi_{T_2}(\text{Fermion}(d))$.
            
            \item $\psi_T(\text{Fermion}(d))=\psi_{T_2}(\text{Fermion}(d))+1$ \quad (in $T$ this is the $dw_2$ edge).

            \item $\psi_T(\text{Fermion}(w_T))=\psi_{T_1}(\text{Fermion}(w_1))+\psi_{T_2}(\text{Fermion}(d))$.

            \item $\psi_T(e)=\psi_{T_1}(e)+\psi_{T_2}(\text{Fermion}(d))$ \quad if $e\in T_1$. 
            
            \item $\psi_T(\text{Fermion}(w_d))=M+\psi_{T_2}(\text{Fermion}(d))+1$.
            
            \item $\psi_T(e)=\psi_{T_2}(e)+M+1$ if $e\in T_2$.
            
        \end{enumerate}
    \end{enumerate}
\end{itemize}

\end{dfn}

\begin{exm}\label{psiorderexm}
Two examples of the $\Psi$-order are given in Figures \ref{psiorderexm44} and \ref{psiorderexm2}. In Figure \ref{psiorderexm3} we give the corresponding chord diagram for the tadpole in Figure \ref{psiorderexm2}. It is worth noticing how the two orders are constructed similarly.

\begin{figure}[h!]
    \centering
    \includegraphics[scale=0.5]{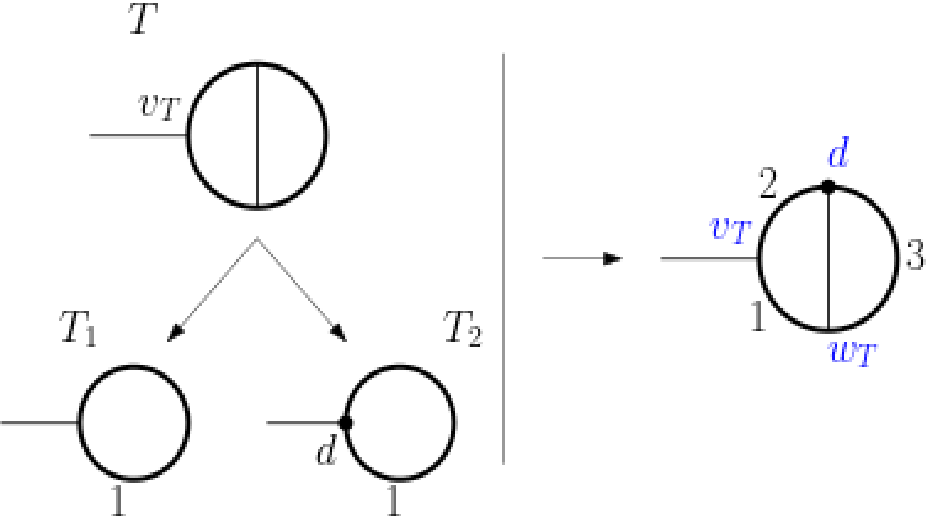}
    \caption{An example of the $\Psi$-order Case 1.}
    \label{psiorderexm44}
\end{figure}

\begin{figure}[h!]
    \centering
    \includegraphics[scale=0.66]{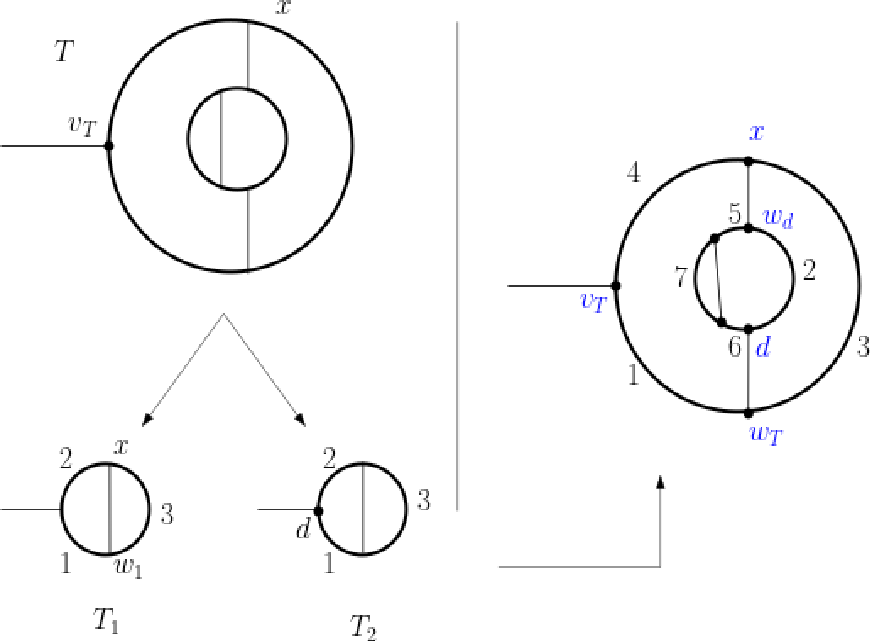}
    \caption{An example of the $\Psi$-order Case 2.}
    \label{psiorderexm2}
\end{figure}
\begin{figure}[h]
    \centering
    \includegraphics[scale=0.7]{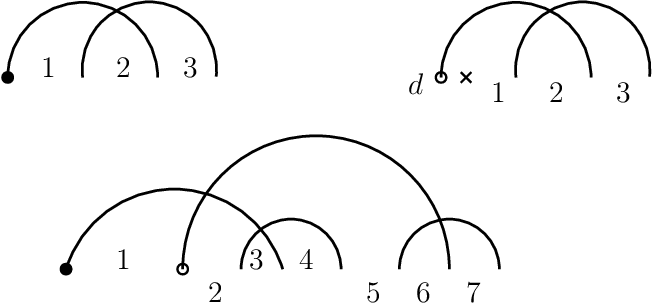}
    \caption{The corresponding chord diagrams for the graphs of Example \ref{psiorderexm}, in Figure \ref{psiorderexm2}.}
    \label{psiorderexm3}
\end{figure}

\end{exm}

Now we can finally express the bijection between Yukawa 1PI tadpoles and connected chord diagrams.

If we use a vertex to indicate an interval in a chord diagram, then we mean, as usual, the interval to the right of the vertex in the linear representation. Analogously we use the fermion edge that comes next to a vertex in counter-clockwise direction. Also note that the interval in the root share decomposition can not be the rightmost interval of the diagram. 

\subsubsection{An Explicit Bijection}

\begin{cor}\label{myresult2inYukawa}
Theorem \ref{myresult1inYukawa} can be used to give an explicit bijection $\Lambda$ between Yukawa theory 1PI tadpoles in (the class $\mathcal{U}_{10}$) and \;$\mathcal{C}$, the class of connected chord diagrams. Namely, $\Lambda: \mathcal{U}_{10}\longrightarrow \mathcal{C}$ is defined recursively as follows:

    \[\Lambda\big(\mathcal{X}\big)\;=\;\raisebox{0cm}{\includegraphics[scale=0.5]{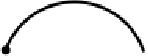}}\;;\qquad\text{and}\quad\]
    \[
    \Lambda(T)=\nabla^{-1}\big(\Lambda(T_1),(\Lambda(T_2),\psi(d)\big),
    \]
    
 where $\Psi^{-1}(T)=(T_1,(T_2,d))$, and, as defined earlier, $\mathcal{X}=\raisebox{-0.25cm}{\includegraphics[scale=0.5]{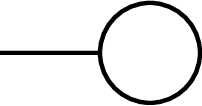}}$.
\end{cor}

\begin{proof}
The proof is straightforward from the definitions of the maps involved. The map is well defined by the uniqueness of the root share decomposition, and is a bijection since $\psi$, $\Psi$, and $\nabla$ are bijections. 
\end{proof}

\begin{exm}\label{completeexmpsi}
Figure \ref{complete} below illustrates all the steps from a tadpole $T$ to its corresponding connected chord diagram. Namely, as the corollary states, given a Yukawa 1PI tadpole $T$, the process can be described as follows:
\begin{enumerate}
    \item Use $\Psi^{-1}$ to decompose $T$ all the way down to copies of $\mathcal{X}$, with extra information about positions $d_i$ at each step
    \item Go up again step by step and use the recursive definition of $\psi$ to order all the fermion edges in the graph.
    \item Again start from the bottom to create the corresponding chord diagrams using $\nabla^{-1}$ and the values $\psi(d)$. The last insertion up gives $\Lambda(T)$.
\end{enumerate}

\begin{figure}[h]
    \centering
    \includegraphics[scale=0.7]{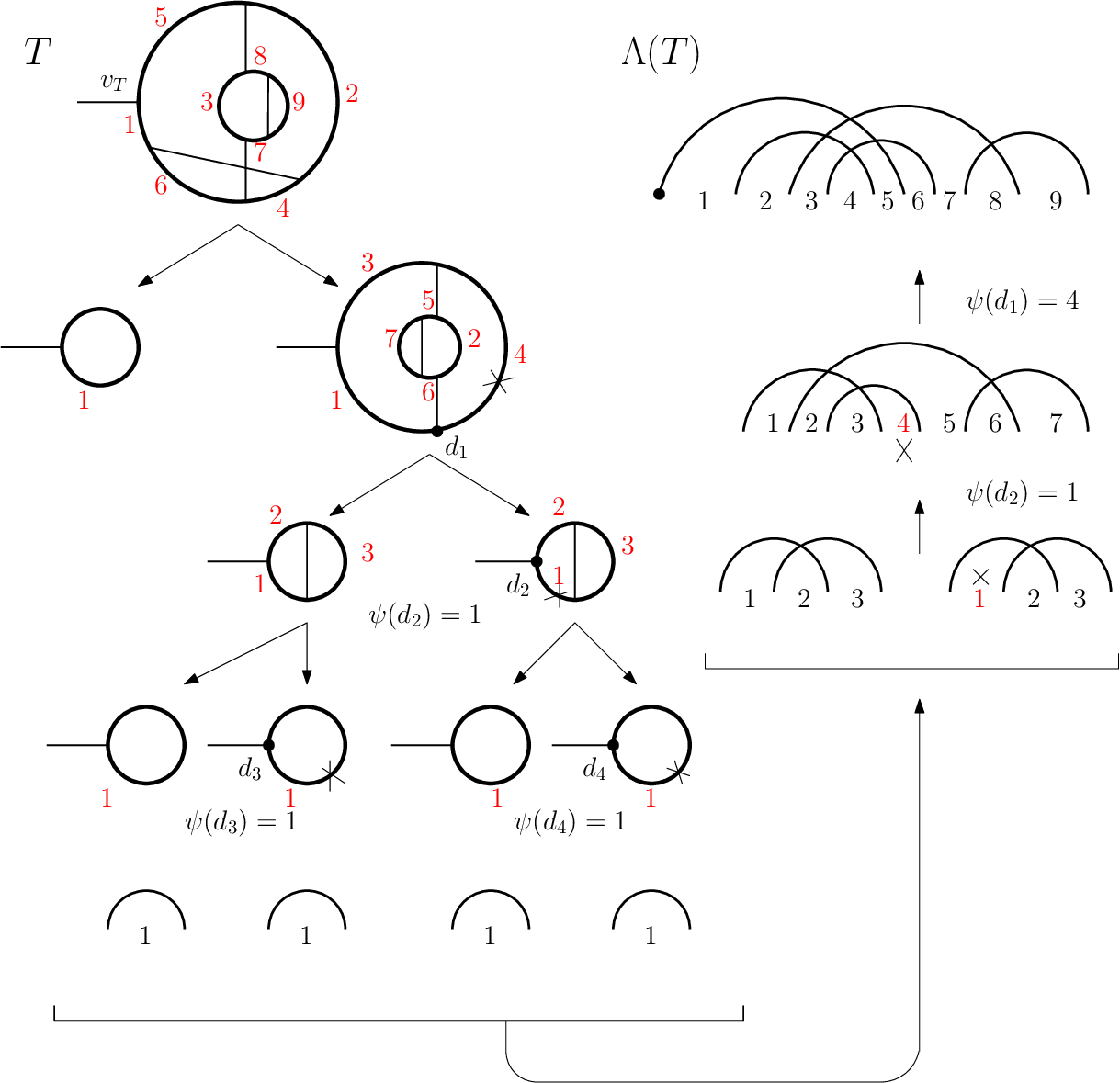}
    \caption{A complete example of the recursive calculation of $\Lambda$.}
    \label{complete}
\end{figure}
\end{exm}

\begin{rem}
It is surprising that, in light of the result in \cite{con}, the bijection between Yukawa 1PI tadpoles and connected chord diagrams gives a bijection between Yukawa 1PI tadpoles and rooted bridgeless combinatorial maps. This will be investigated in future work.
\end{rem}


\newpage\subsection{Yukawa Vacuum Graphs:\quad $\left.\partial_{\phi_c}^0(\partial_{\psi_c}\partial_{\bar{\psi}_c})^0G^{\text{Yuk}}(\hbar,\phi_c,\psi_c)\right|_{\phi_c=\psi_c=0}$}

Here we interpret line 1 in Table \ref{table4}. By definition, $\left.\partial_{\phi_c}^0(\partial_{\psi_c}\partial_{\bar{\psi}_c})^0G^{\text{Yuk}}(\hbar,\phi_c,\psi_c)\right|_{\phi_c=\psi_c=0}$ generates all 1PI Yukawa graphs with no external legs. In physics jargon these are called vacuum graphs. 

\begin{figure}[h]
    \centering
    \includegraphics[scale=0.75]{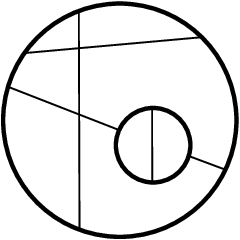}
    \caption{A 1PI vacuum graph in Yukawa theory.}
    \label{vaccc}
\end{figure}

By Lemma \ref{property1Yukawa} we know that for a vaccuum graph $\Gamma$ we will still have $|V(\Gamma)|=f$, the number of internal fermion edges. Consequently, we also still get $p=\ell(\Gamma)-1$ and $|V(\Gamma)|=2p$  (Euler's formula for the first), where $p$ is the number of internal boson edges. Let $V(x)$ be the generating series of 1PI Yukawa vacuum graphs counted by the number of boson edges.

Before giving the chord-diagrammatic interpretation notice that in Table \ref{table4} we have $$[\hbar^0]\left.\partial_{\phi_c}^0(\partial_{\psi_c}\partial_{\bar{\psi}_c})^0G^{\text{Yuk}}(\hbar,\phi_c,\psi_c)\right|_{\phi_c=\psi_c=0}=[\hbar^1]\left.\partial_{\phi_c}^0(\partial_{\psi_c}\partial_{\bar{\psi}_c})^0G^{\text{Yuk}}(\hbar,\phi_c,\psi_c)\right|_{\phi_c=\psi_c=0}=0,$$ since we do not consider the empty graph or the plain loop to be 1PI graphs.  The following proposition proves this conjecture and is a consequence of Theorem \ref{myresult1inYukawa}.

\begin{prop}\label{myresultykawavacuum}
 Let $V(x)$ be the generating series of 1PI Yukawa vacuum graphs counted by the number of boson edges. Then 
 \[V(x)=\displaystyle\frac{C(x)^2}{2x},\]
which implies that \[[\hbar^{n+1}]\left.\partial_{\phi_c}^0(\partial_{\psi_c}\partial_{\bar{\psi}_c})^0G^{\text{Yuk}}(\hbar,\phi_c,\psi_c)\right|_{\phi_c=\psi_c=0}=[x^n]\displaystyle\frac{C(x)^2}{2x}.\]
\end{prop}
\begin{proof}
Let $\mathcal{U}_{00}$ be the class of Yukawa 1PI vacuum graphs. It is clear that $\mathcal{U}_{10}-\{\mathcal{X}\}=X\ast\mathcal{U}_{00}^\bullet$, where $\mathcal{U}_{00}^\bullet$ is the class of vacuum graphs with a distinguished fermion edge (or equivalently, with a distinguished vertex), and $X$ as usual is used for the single constituent whose generating function is $x$ (in this case it refers to an external boson edge to be inserted). Also note that our vacuum graphs must have at least one boson edge (the plain fermion loop is not considered 1PI).

Since there are two vertices for each chosen boson edge, we have that the generating function for $\mathcal{U}_{00}^\bullet$ is 
\[2xV^\prime(x).\]
Thus, by Theorem \ref{myresult1inYukawa}, we should have 
\[C(x)-x=2x^2V^\prime(x).\] 
Then, by Lemma \ref{cd} we know that $2xCC'=C^2+C-x$, and we get
\[
  V^\prime(x)=\displaystyle\frac{1}{2x^2}(C(x)-x)
    =\displaystyle\frac{1}{2x^2}(2xC(x)C^\prime(x)-C(x)^2)
    =\displaystyle\frac{1}{2}\frac{(x(C(x)^2)^\prime-C(x)^2)}{x^2},
\]
which can then be integrated to give the result.
\end{proof}


\subsection{Yukawa Graphs from $\left.\partial_{\phi_c}^2(\partial_{\psi_c}\partial_{\bar{\psi}_c})^0G^{\text{Yuk}}(\hbar,\phi_c,\psi_c)\right|_{\phi_c=\psi_c=0}$}

The Yukawa 1PI graphs counted by loop number in line 3 of Table \ref{table4} are the graphs with exactly two external legs, each of which are boson-type. Again, for any such graph we have by Lemma \ref{property1Yukawa} and Euler's formula that $|V(\Gamma)|=f$, $p=\ell(\Gamma)-1$ and $|V(\Gamma)|=2p$;
 where $p$ is the number of internal boson edges and $f$ is the number of internal fermion edges. 
 
 By their definition, we see that these graphs are simply tadpoles with a distinguished fermion edge at which a \textit{second} external boson leg is inserted. 
 
 \begin{rem}\label{importremark}
 By the word `second' above we literally mean that the roles of the two boson edges are physically different. This will be reflected in that we will always assume that one boson leg is the `left' or `first' one. For example, the graphs 
 \begin{center}
     \raisebox{0cm}{\includegraphics[scale=0.87]{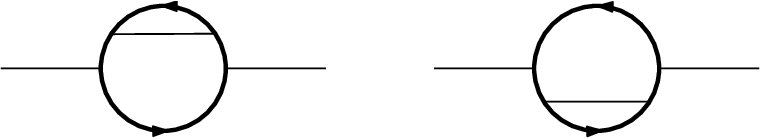}}
 \end{center}
 are considered different even though one of them can be rotated to get the other one.
 \end{rem}
 
 However, in  the process of distinguishing a fermion edge of a tadpole, we have to exclude the fermion edge immediately before the tadpole's leg vertex as it will yield the same graph if the next fermion edge is chosen instead. Thus the generating function of these graphs, counted by the number of all boson edges is given by
 \begin{equation}\label{myvacuumeq1}
     U_{20}(x)=x(2xT^\prime(x)-T(x))=x(2xC'(x)-C(x)),
 \end{equation} where $T$ is as in Theorem \ref{myresult1inYukawa}, the generating function for Yukawa 1PI tadpoles counted by the number of all boson edges. 
 \begin{figure}[h]
     \centering
     \includegraphics[scale=0.7]{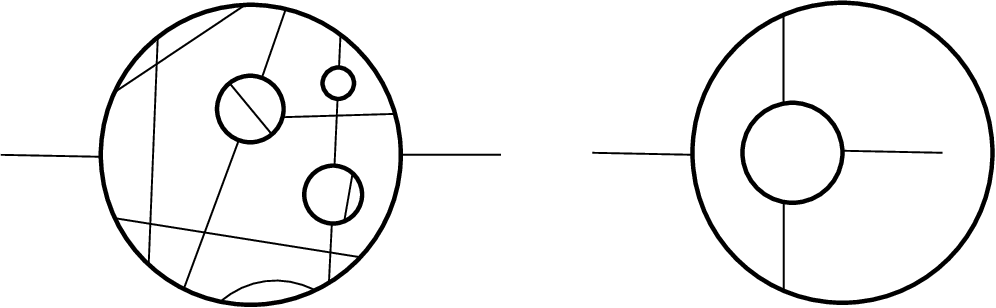}
     \caption{Two graphs generated by $\left.\partial_{\phi_c}^2(\partial_{\psi_c}\partial_{\bar{\psi}_c})^0G^{\text{Yuk}}(\hbar,\phi_c,\psi_c)\right|_{\phi_c=\psi_c=0}$.}
     \label{twoleg}
 \end{figure}
 
\begin{rem}
 Notice that the two boson legs are not necessarily on the same fermion loop, see the second graph in Figure \ref{twoleg} for example.
\end{rem}

Thus, equation \ref{myvacuumeq1} shows that 
\begin{align}\label{vacuumeq}
    [\hbar^{n}]\left.\partial_{\phi_c}^2(\partial_{\psi_c}\partial_{\bar{\psi}_c})^0G^{\text{Yuk}}(\hbar,\phi_c,\psi_c)\right|_{\phi_c=\psi_c=0}&=[x^{n+1}]U_{20}(x)=\;x(2xT^\prime(x)-T(x))\\&=[x^n](2xC'(x)-C(x)).
\end{align} Further, by Proposition \ref{myproposition2connected}, it follows that  
\begin{align}\label{myequ1}
    [\hbar^{n}]\left.\partial_{\phi_c}^2(\partial_{\psi_c}\partial_{\bar{\psi}_c})^0G^{\text{Yuk}}(\hbar,\phi_c,\psi_c)\right|_{\phi_c=\psi_c=0}&=[x^n] \;\;\displaystyle\frac{C(x)^2}{x} \;\left[\left. \displaystyle\frac{C_{\geq2}(t)}{t^2}\right|_{t=C(x)^2/x}\right], \;\text{or equivalently}\nonumber\\
    U_{20}(x)&=C(x)^2 \;\left[\left. \displaystyle\frac{C_{\geq2}(t)}{t^2}\right|_{t=C(x)^2/x}\right].
\end{align}

This equation will be useful in providing a chord-diagrammatic interpretation for graphs generated by $\left.\partial_{\phi_c}^1(\partial_{\psi_c}\partial_{\bar{\psi}_c})^1G^{\text{Yuk}}(\hbar,\phi_c,\psi_c)\right|_{\phi_c=\psi_c=0}$, as we shall see next. For the next section we will need to give a name for the class of graphs considered here, and, as it became our good habit, we shall denote it by $\mathcal{U}_{20}$.

\subsection{Yukawa Graphs from $\left.\partial_{\phi_c}^1(\partial_{\psi_c}\partial_{\bar{\psi}_c})^1G^{\text{Yuk}}(\hbar,\phi_c,\psi_c)\right|_{\phi_c=\psi_c=0}$}
These are the Yukawa 1PI graphs with vertex-type residue. Line 5 in Table \ref{table4} gives the number of these graphs, sized with loop number, up to size 5. These have two external fermion legs in addition to one boson leg (see Figure \ref{exaaample1Y}).
\begin{figure}[h]
     \centering
     \includegraphics[scale=0.6]{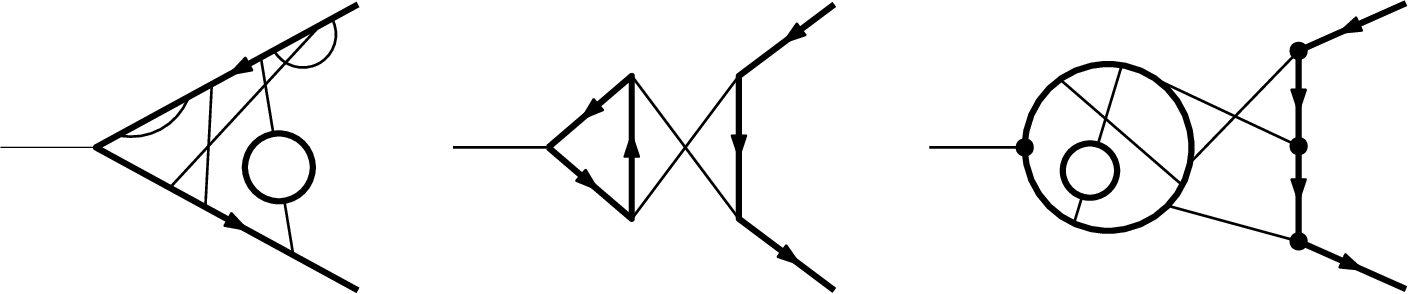}
     \caption{Examples of graphs generated by $\left.\partial_{\phi_c}^1(\partial_{\psi_c}\partial_{\bar{\psi}_c})^1G^{\text{Yuk}}(\hbar,\phi_c,\psi_c)\right|_{\phi_c=\psi_c=0}$.}
     \label{exaaample1Y}
 \end{figure}
Notice that if the ends of the two fermion external legs were identified we still wouldn't get a general tadpole. The reason for this is that, since the graph is  1PI, we can not have something like the one in Figure \ref{forboddengraph} below.
\begin{figure}[h]
    \centering
    \raisebox{0cm}{\includegraphics[scale=0.54]{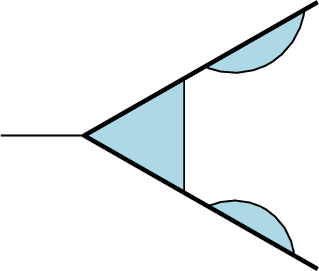}}
    
    \caption{A forbidden graph.}
    \label{forboddengraph}
\end{figure}
Knowing this is exactly the way we get the chord-diagrammatic interpretation. We will let $\mathcal{U}_{11}$ denote the class of al graphs generated by $\left.\partial_{\phi_c}^1(\partial_{\psi_c}\partial_{\bar{\psi}_c})^1G^{\text{Yuk}}(\hbar,\phi_c,\psi_c)\right|_{\phi_c=\psi_c=0}$. Before proceeding to the next theorem recall that by Lemma \ref{property1Yukawa} we have that for any $\Gamma\in\mathcal{U}_{11}$ 
$|V(\Gamma)|=f+1$, and hence by Euler's formula $p=\ell(\Gamma)$, where $f$ (and $p$) is the number of internal fermion (boson) edges. 

\begin{nota}\label{alinota1}
In representing the graphs in $\mathcal{U}_{11}$, we still stick to the counter-clockwise convention, even for the unique path of fermion edges. We shall always represent the graphs in $\mathcal{U}_{11}$ by fixing the boson external leg to the left, and then the fermion external half-edge directed towards the boson-leg vertex will be called the \textit{upper end} and will be denoted with $u_\Gamma$; on the other hand the fermion external half-edge directed away from the boson-leg vertex will be called the \textit{lower end} and will be denoted by $l_\Gamma$.
\end{nota}
\begin{center}
\raisebox{0cm}{\includegraphics[scale=0.6]{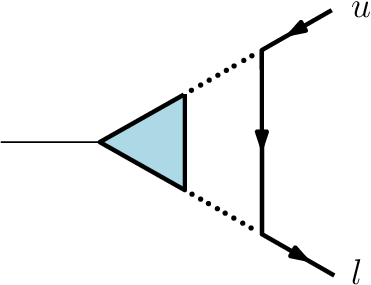}}
\end{center}

\begin{thm}
\label{mypropyukawa3}\label{ali}
 Let $\mathcal{U}_{11}$ be the class of Yukawa 1PI graphs (with vertex-type residue) generated by $\left.\partial_{\phi_c}^1(\partial_{\psi_c}\partial_{\bar{\psi}_c})^1G^{\text{Yuk}}(\hbar,\phi_c,\psi_c)\right|_{\phi_c=\psi_c=0}$, and let $U_{11}(x)$ be the generating series of these graphs, counted by the number of all boson edges. Then  
 
 \[U_{11}(x)=x\;\left. \displaystyle\frac{C_{\geq2}(t)}{t^2}\right|_{t=C(x)^2/x}.\]
\end{thm}

\begin{proof}
We start with another type of graphs, namely with the class $\mathcal{U}_{20}$ of the previous section. We will describe a bijection 
\[M:\mathcal{U}_{11}\ast (\mathcal{U}_{10}\ast\mathcal{U}_{10})\longrightarrow \mathcal{X}\ast\mathcal{U}_{20}.\]
The construction is simple:
Assume $(\Gamma,(T_1,T_2))$ is a triplet from $\mathcal{U}_{11}\ast (\mathcal{U}_{10}\ast\mathcal{U}_{10})$. Let $u$ and $l$ be the upper and lower ends of $\Gamma$ as described in Notation \ref{alinota1}. Now take the tadpoles $T_1$ and $T_2$ and do the following:
\begin{enumerate}
    \item For $T_1$, let $v_1$ be the vertex at the boson leg as before. Let $f_1$ be the fermion edge immediately before $v_1$ on Loop$(v_1)$. Detach $f_1$ from $v_1$ and denote the unique resulting graph with $u(T_1)$. This can be depicted as in Figure \ref{bell1fig} below.
    \begin{figure}[h]
        \centering
       \raisebox{0cm}{\includegraphics[scale=0.73]{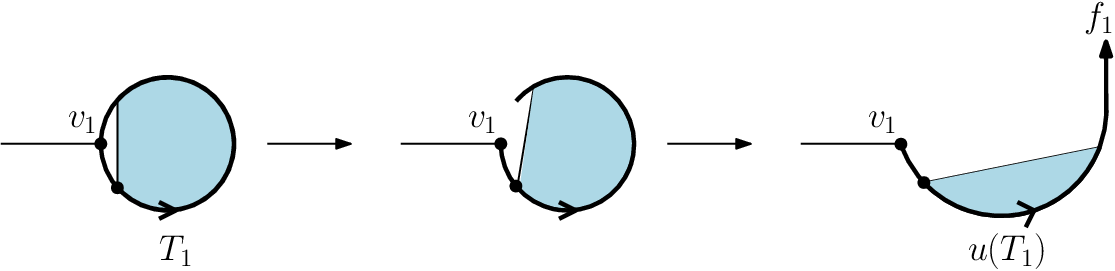}}
       \caption{$u(T_1)$}
        \label{bell1fig}
    \end{figure}
    
    \item For $T_2$, let $v_2$ be the vertex at the boson leg as before. Let $f_2$ be Fermion$(v_2)$, the fermion edge immediately next to $v_2$ on Loop$(v_2)$. Detach $f_2$ from $v_2$ and denote the unique resulting graph with $l(T_2)$. This can be depicted as in Figure \ref{bell2fig} below.
    \begin{figure}[h]
        \centering
       \raisebox{0cm}{\includegraphics[scale=0.73]{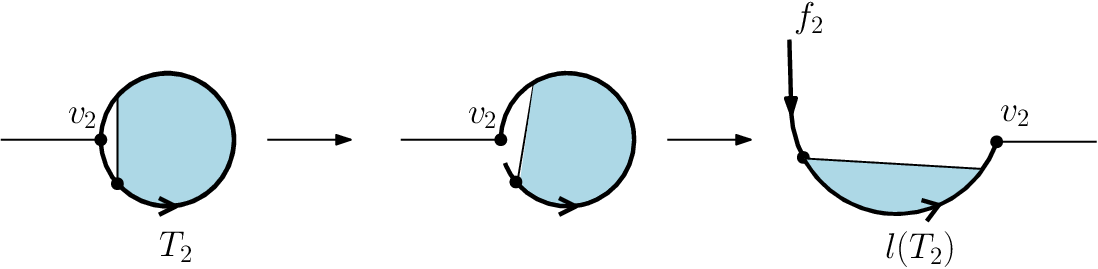}}
       \caption{$l(T_2)$}
        \label{bell2fig}
    \end{figure}
    
    \item Identify $f_1$ with $u$ of $\Gamma$ and identify $f_2$ with $l$. Denote the resulting graph with $\tilde{\Gamma}$
    \item Identify the vertices $v_1$ and $v_2$ in $\tilde{\Gamma}$.
    \item By removing one of the external boson legs we get $M(\Gamma, T_1,T_2)$.
    The process described here can be depicted as in Figure \ref{figm} below.
    \begin{figure}[h]
        \centering
       \raisebox{0cm}{\includegraphics[scale=0.7]{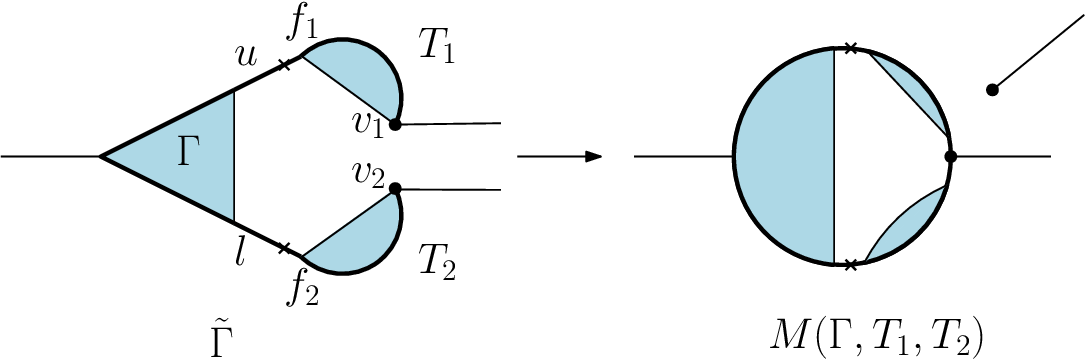}}
       \caption{$M(\Gamma, T_1,T_2)$}
        \label{figm}
    \end{figure}
\end{enumerate}

By our construction, and by the uniqueness of the representation of graphs described in Notation \ref{alinota1}, the map $M$ described this way is well-defined. Moreover, it is reversible: Assume we are given a graph $G$ in $\mathcal{U}_{20}$ in canonical representation (i.e. drawn according to our counter-clockwise convention). Remember from Remark \ref{importremark}, that the representation of $G$ is unique and one boson leg is the `left' external leg. Let $w_l$ be and $w_r$ be the corresponding vertices at the left and right  boson legs respectively. Let us refer to the half fermion loop from $w_r$ to $w_l$ by the \textit{upper half loop}, and similarly we will refer to the half fermion loop from $w_l$ to $w_r$ by the \textit{lower half loop} (see Figure \ref{notaM}).  Now we do the following:
\begin{enumerate}
    \item Starting from $w_r$, remove Fermion$(w_r)$ from the graph $G$, and move along the rest of the upper half loop (direction is as before) searching for the \underline{\textbf{last}} fermion edge whose removal disconnects $G-$Fermion$(w_r)$. Stop the search at $w_l$. Notice that it may happen that no such edge exists. If found, denote this edge by $f_1$. 
    If such an edge doesn't exist we set $f_1=$ Fermion$(w_r)$.  
    
    \item Similarly, starting from $w_l$, remove Fermion$(w_l)$ from the graph $G$, and move along the rest of the lower half loop searching for the \underline{\textbf{first}} fermion edge whose removal disconnects $G-$Fermion$(w_l)$. Stop the search at $w_r$. Notice that it may happen that no such edge exists. If found, denote this edge by $f_2$. 
    If such an edge doesn't exist we set $f_2 =$ (fermion edge immediately before $w_r$).
    
    \item Cut at $f_1$ and $f_2$, and identify their ends with $w_r$ such that the direction of $f_1$ with respect to $w_r$ is counter-clockwise, whereas the direction of $f_2$ with respect to $w_r$ is to be made clockwise.
    
    \item Obtain $T_1$ and $T_2$ by splitting $w_r$ and its boson leg into two copies. The remnant of $G$ is $\Gamma$. 
    
    \begin{figure}[h]
        \centering
       \raisebox{0cm}{\includegraphics[scale=0.577]{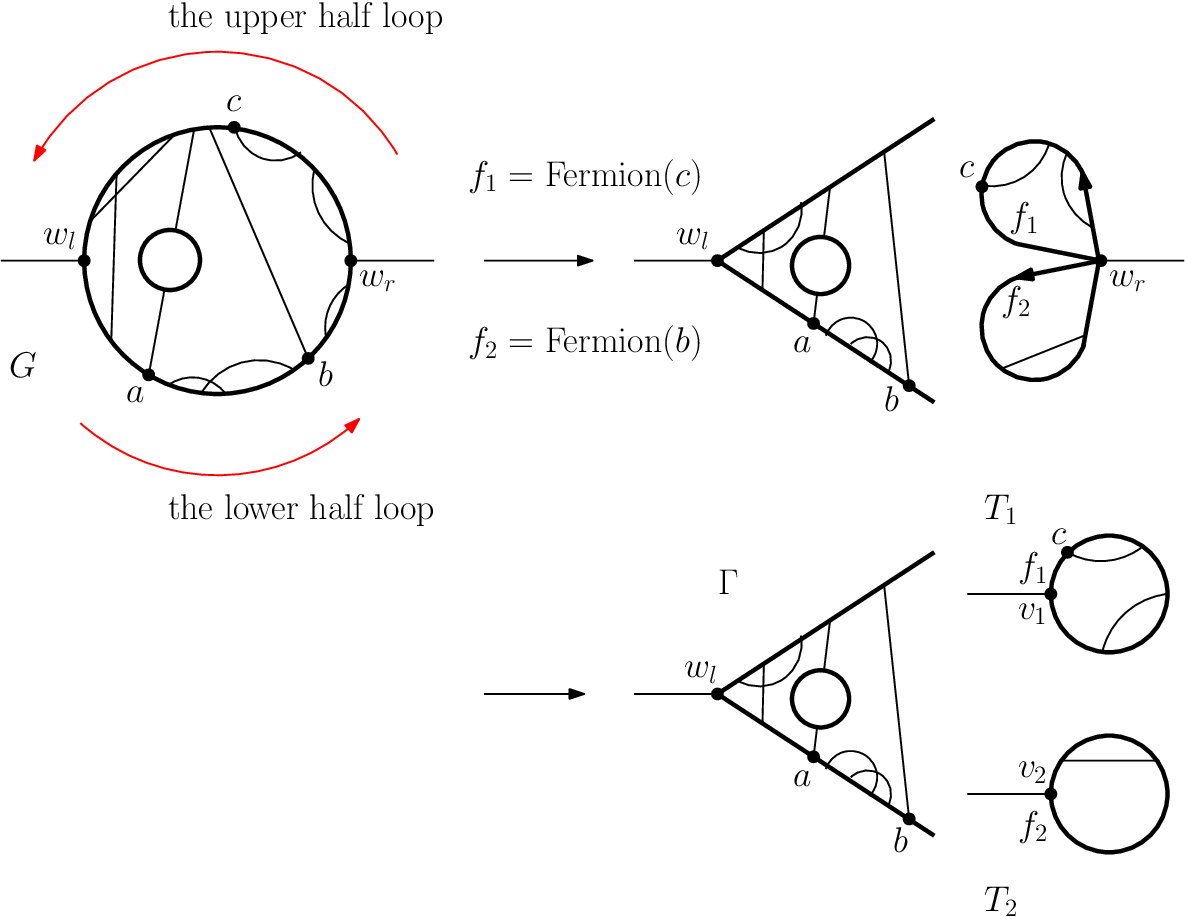}}
       \caption{Calculating $M^{-1}(G)$.}
        \label{notaM}
    \end{figure}
\end{enumerate}

It is worth noting that in Figure \ref{notaM} $f_2$ is Fermion$(b)$ and is not Fermion$(a)$ as can be checked using the definition.

Thus, the map $M:\mathcal{U}_{11}\ast (\mathcal{U}_{10}\ast\mathcal{U}_{10})\longrightarrow \mathcal{X}\ast\mathcal{U}_{20}$ is a bijection. Consequently, on the level of generating functions we will have 
\[
    U_{11}(x) T(x)^2= x \;U_{20}(x),
\]

and hence, by equation (\ref{myequ1}) and Theorem \ref{myresult1inYukawa}, we get  
\[U_{11}(x) C(x)^2=x\;C(x)^2 \;\left[\left. \displaystyle\frac{C_{\geq2}(t)}{t^2}\right|_{t=C(x)^2/x}\right],\]
that is,
\begin{equation}\label{myeq2}
    U_{11}(x)=\left. \displaystyle\frac{C_{\geq2}(t)}{t^2}\right|_{t=C(x)^2/x}, 
\end{equation}
as desired. 
\end{proof}

\subsection{Yukawa  Graphs from $\left.\partial_{\phi_c}^0(\partial_{\psi_c}\partial_{\bar{\psi}_c})^1G^{\text{Yuk}}(\hbar,\phi_c,\psi_c)\right|_{\phi_c=\psi_c=0}$}

Let $\mathcal{U}_{01}$ be the class of Yukawa 1PI graphs generated by $\left.\partial_{\phi_c}^0(\partial_{\psi_c}\partial_{\bar{\psi}_c})^1G^{\text{Yuk}}(\hbar,\phi_c,\psi_c)\right|_{\phi_c=\psi_c=0}$. Line 4 in Table \ref{table4} gives the number of these graphs, sized with loop number, up to size 5. These are graphs with two external fermion legs and with no  boson leg. 

 For any $\Gamma\in\mathcal{U}_{01}$, we have by Lemma \ref{property1Yukawa} that
$|V(\Gamma)|=f+1$, and hence by Euler's formula $p=\ell(\Gamma)$, where $f$ (and $p$) is the number of internal fermion (boson) edges. In particular, 
\begin{equation}\label{numinternalfermion}
    f=2p-1.
\end{equation}
We let $U_{01}(x)$ be the generating function of the graphs in $\mathcal{U}_{01}$ counted by the number of boson edges.

There are two ways to think about this type of graphs, both can be used to obtain $U_{01}(x)$:
\begin{enumerate}
    \item It is clear that the graphs in $\mathcal{U}_{11}$, which were considered in the last section, are the rooted versions of the graphs in $\mathcal{U}_{01}$, namely,  except for the one-vertex graph, every graph in $\mathcal{U}_{11}$  is obtained from a unique graph in $\mathcal{U}_{01}$ by distinguishing an internal fermion edge and inserting a boson leg. This means that, on the level of generating functions:

\begin{equation}\label{letsdoit}
    U_{11}(x)=x\;(2x\;U_{01}^\prime(x)-U_{01}(x)+1),
\end{equation}where the RHS uses the fact that $f=2p-1$ by equation (\ref{numinternalfermion}), and where the $1$ corresponds to the one-vertex graph graph in $\mathcal{U}_{11}$.

\item The other way is to think of a tadpole as being constructed from a list of graphs from $\mathcal{U}_{01}$. This way is more direct and we shall discuss it below.
\end{enumerate}

\begin{prop}
 A Yukawa 1PI tadpole graph can be decomposed as  boson leg together with a list of graphs from $\mathcal{U}_{01}$. In particular, on the level of generating functions we will have
 \begin{equation}\label{tobeusd}
     T(x)=\displaystyle\frac{x}{1-U_{01}(x)}.
 \end{equation}\label{tobeused2}
 This in turn implies that \begin{equation}
     U_{01}(x)=2xC^\prime(x)-C(x)=C(x)^2 \;\left[\left. \displaystyle\frac{C_{\geq2}(t)}{t^2}\right|_{t=C(x)^2/x}\right]=U_{20}(x).
 \end{equation}
\end{prop}

 \begin{proof}
 Given a tadpole $T\in\mathcal{U}_{10}$, we will give a unique decomposition into a boson leg together with  a list of graphs from $\mathcal{U}_{01}$. Let $v_T$ be the vertex at the boson external leg of $T$ as usual. Also let $f_0$ be the fermion edge on Loop$(v_T)$ immediately before $v_T$. 
 
 \begin{enumerate}
     \item If $f_0=$ Fermion$(v_T)$, return $(\raisebox{-0.cm}{\raisebox{-0.05cm}{\includegraphics[scale=0.5]{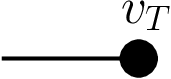}}};\varnothing)$.
     \item  Otherwise, detach $f_0$ and call the resulting graph $u(T)$. \item Starting from $v_T$, move along Loop$(v_T)$ and determine all fermion edges on Loop$(v_T)$ whose removal disconnects $u(T)$. Let this list of fermion edges be 
     \[f_1,f_2,\ldots,f_n.\]
     (Note that it will always be the case that $f_1=$ Fermion$(v_T)$).
     \item For $1\leq i\leq n-1$ define $G_i$ to be the graph obtained from the original $T$ by 
     
     \begin{itemize}
         \item cutting at $f_i$ and $f_{i+1}$, and
         \item deleting the component that contains $v_T$.
     \end{itemize}
     
     \item For $i=n$, define $G_n$ to be the graph obtained from the original $T$ by 
     
     \begin{itemize}
         \item cutting at $f_n$ and $f_{0}$, and
         \item deleting the component that contains $v_T$.
     \end{itemize}
     \item Return $(\raisebox{-0.05cm}{\includegraphics[scale=0.5]{Figures/vt.eps}};G_1,G_2,\ldots,G_n)$. 
 \end{enumerate}
 
 It is easily seen that the graphs $G_i$ are 2-edge connected, in fact the graphs $G_i$ are maximal 1PI's inserted along Loop$(v_T)$. Besides, by their definition, the graphs $G_i$ have a fermion-type residue and are therefore in $\mathcal{U}_{01}$. Moreover, this construction is clearly unique for every tadpole in  $\mathcal{U}_{10}$. 
 
 Conversely, every such list can be uniquely used to produce a tadpole. Figure \ref{Gs} below illustrates two cases of this decomposition.
 \begin{figure}[h]
        \centering
       \raisebox{0cm}{\includegraphics[scale=0.68]{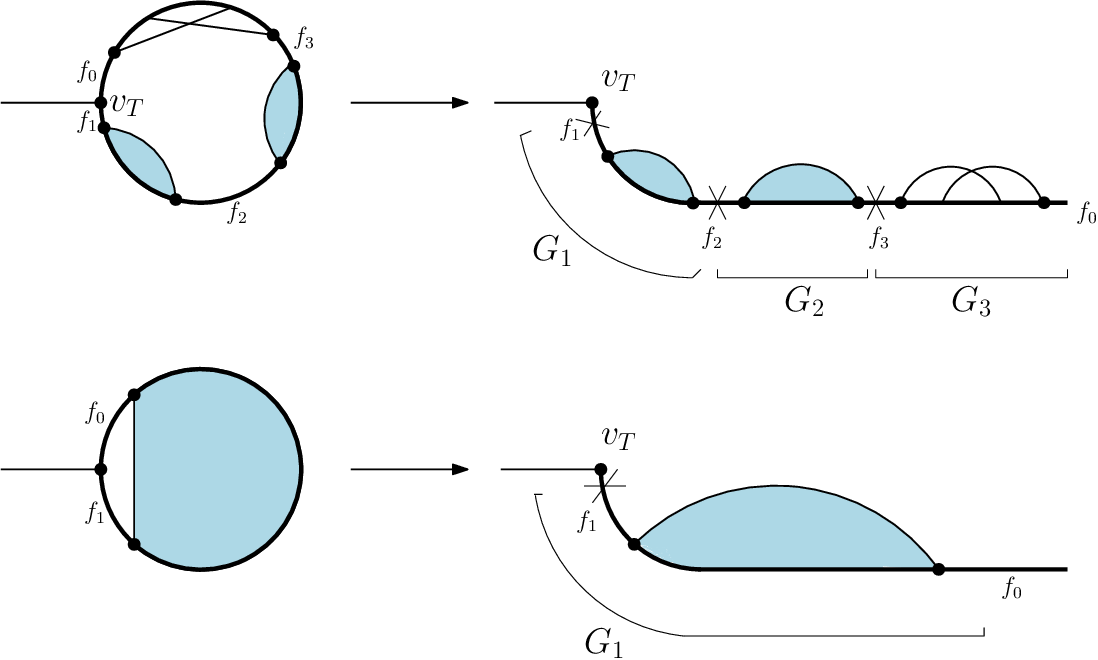}}
       \caption{Examples of the decomposition of tadpoles into a list of graphs from $\mathcal{U}_{01}$.}
        \label{Gs}
    \end{figure}

This gives that, on the level of generating functions,  
  \[T(x)=\displaystyle\frac{x}{1-U_{01}(x)}.\]
 
 By Lemma \ref{cd} recall that $C(x)=\frac{x}{1-(2xC'(x)-C(x))}$, and since we have $T(x)=C(x)$ we get 
     \[U_{01}(x)=2xC^\prime(x)-C(x).\]
 Then it follows from Proposition \ref{myproposition2connected} that  
     
     \[U_{01}(x)=C(x)^2 \;\left[\left. \displaystyle\frac{C_{\geq2}(t)}{t^2}\right|_{t=C(x)^2/x}\right]=U_{20}(x).\]
 \end{proof}


\textbf{Conclusion:} We have thus obtained a chord-diagrammatic interpretation for a number of proper Green functions in Yukawa theorey and quenched QED. In \cite{michiq}, the asymptotics of these generating series are obtained by means of singularity analysis. Here, having obtained every generating series in terms of connected chord diagrams, the task of obtaining the asymptotics is readily done and is a straightforward calculation now. Indeed, we only need to use our knowledge of $\mathcal{A}_{\frac{1}{2}}^2C(x)$ (and $\mathcal{A}_{\frac{1}{2}}^2C_{\geq2}(x)$) to obtain the asymptotics for the different Yukawa and QQED green functions considered here.

\bibliographystyle{plain}
\bibliography{References11.bib}

\begin{thebibliography}{10}

\bibitem{michi1}
M.~Borinsky.
\newblock {\em Generating asymptotics for factorially divergent sequences}.
\newblock arXiv preprint arXiv:1603.01236, 2016.

\bibitem{michiq}
M.~Borinsky.
\newblock {\em Renormalized asymptotic enumeration of Feynman diagrams}.
\newblock Annals Phys. 385 (2017) 95-135, DOI: 10.1016/j.aop.2017.07.009, 2017.

\bibitem{michi}
M.~Borinsky.
\newblock {\em Graphs in perturbation theory: Algebraic structure and
  asymptotics}.
\newblock arXiv:1807.02046, 2018.

\bibitem{michilattice}
M.~Borinsky.
\newblock {\em Algebraic lattices in QFT renormalization}.
\newblock Letters in Mathematical Physics, Volume 106, Issue 7, pp 879-911,
  arXiv:1509.01862, July 2016,.

\bibitem{conneskreimer}
A.~Connes and D.~Kreimer.
\newblock {\em Renormalization in quantum field theory and the
  Riemann–Hilbert problem i: The Hopf algebra structure of graphs and the
  main theorem}.
\newblock Communications in Mathematical Physics, 210:249–273, 2000.

\bibitem{con}
J.~Courtiel, K.~Yeats, and N.~Zeilberger.
\newblock {\em Connected chord diagrams and bridgeless maps}.
\newblock arXiv:1611.04611, 2017.

\bibitem{numandweights}
P.~Cvitanovi\'c, B.~Lautrup, and R.~B. Pearson.
\newblock {\em Number and weights of Feynman diagrams}.
\newblock Phys. Rev. D Vol. 18 (6), 1978.

\bibitem{flajoletchords}
P.~Flajolet and M.~Noy.
\newblock {\em Formal Power Series and Algebraic Combinatorics: 12th
  International Conference, FPSAC’00, Moscow, Russia, June 2000,
  Proceedings}.
\newblock Journal of Algorithms, pp. 191–201, Berlin, Heidelberg: Springer
  Berlin Heidelberg, 2000.

\bibitem{kjm2}
D.~M. Jackson, A.~Kempf, and A.~Morales.
\newblock {\em Algebraic Combinatorial fourier and Legendre Transforms with
  Application in Perturbative Quantum Field Theory}.
\newblock arXiv:1805.09812v3, 2019.

\bibitem{alipaper2con}
A.~A. Mahmoud.
\newblock {\em An Asymptotic Expansion for the Number of 2-Connected Chord
  Diagrams}.
\newblock Journal of Mathematical Physics (Vol.64, Issue 12), DOI:
  10.1063/5.0171074, 2023.

\bibitem{manchonhopf}
D.~Manchon.
\newblock {\em Hopf algebras, from basics to applications to renormalization}.
\newblock arXiv preprint math/0408405, 2004.

\end{thebibliography}

\end{document}